\documentclass[journal,onecolumn,12pt]{IEEEtran}
\usepackage{epsfig,latexsym}
\usepackage{float}
 \usepackage{enumerate}
\usepackage{indentfirst}
\usepackage{amsmath}
\usepackage{amssymb}
\usepackage{times}
\usepackage{subfigure}
\usepackage{cite}
\usepackage{url}
\usepackage{color}

\definecolor{green}{rgb}{0.796,0.948,0.816}
\newtheorem{Definition}{Definition}

\newtheorem{Proposition}{Proposition}

\newtheorem{Lemma}{Lemma}

\newtheorem{Remark}{Remark}

\newtheorem{theorem}{$\mathbf{Theorem}$}

\begin{document}
{\renewcommand\baselinestretch{1.2}\selectfont
\title{On the Outage Performance of Non-Orthogonal Multiple Access with  One-Bit
Feedback}
\author{Peng Xu, Yi Yuan, Zhiguo Ding, \IEEEmembership{Senior Member, IEEE},
 Xuchu Dai, \\and Robert Schober, \IEEEmembership{Fellow, IEEE}
\thanks{Peng Xu and Xuchu Dai are with Dept. of Electronic Engineering and Information Science, University of Science and Technology of China, P.O.Box No.4,
230027, Hefei, Anhui, China. Yi Yuan and Zhiguo Ding are with  School of
Computing and Communications, Lancaster University, LA1 4WA, UK.
 Robert Schober is with Institute for Digital Communications, the Friedrich-Alexander University (FAU), Erlangen 91058, Germany.
}\vspace{-2em}}\maketitle\par}

\begin{abstract}
 In this paper, the outage performance of downlink non-orthogonal multiple access (NOMA)
  is  investigated for the case
 where each user feeds back only one bit of its  channel state information (CSI) to the base station. Conventionally,
  opportunistic  one-bit feedback has been used in  fading broadcast channels to select  only one
   user   for   transmission. In contrast,
 the considered NOMA scheme adopts superposition coding  to serve all  users simultaneously
  in order to improve  user fairness.
 A closed-form expression for the common outage probability (COP)
 is derived, along with the optimal diversity gains under two types of power constraints. Particularly, it is demonstrated  that the diversity gain
 under a long-term power constraint is twice as large as that under a short-term power constraint.
 Furthermore, we study  dynamic power allocation optimization for minimizing the COP, based on  one-bit CSI feedback.
 This problem is challenging since the objective function is non-convex; however, under the short-term power constraint, we demonstrate that the original  problem can be transformed
 into a set  of convex problems. Under the long-term power constraint,
  an asymptotically optimal solution is obtained for high signal-to-noise ratio.

\end{abstract}
\vspace{-1em}
\section{Introduction}
Non-orthogonal multiple access (NOMA) has been recognized
 as an important   multiple access (MA) technique  in  future fifth generation (5G) networks
 since a balanced tradeoff between spectral efficiency and user fairness can be realized   \cite{saito2013system,li20145g,ding2014performance,timotheou2015fairness,
 liu2015cooperative,Ding2015cooperative_NOMA,shi2015outage,yang2016outage}.
 Unlike conventional MA, such as
 time-division multiple access (TDMA), NOMA simultaneously transmits messages to multiple users.
 The power domain is utilized by NOMA such that different users are
 served at different power levels.
 The basic idea of NOMA is motivated by the optimal coding scheme for
  the broadcast channel (BC) \cite{cover1972broadcast}, which  combines   superposition coding at the transmitter with
 successive interference cancellation (SIC) decoding at the  receivers.
 However, compared to the conventional transmission schemes for the BC, NOMA imposes an additional fairness
 constraint on   transmission, i.e.,  more power is always allocated
  to the users with poorer channel conditions, which is  different from the
 conventional waterfilling power allocation scheme.
  In this sense, NOMA can be viewed as a special case of the superposition
coding developed for the BC \cite{xu2015new}.

The capacity region of the degraded discrete
memoryless BC was first found by Cover based on superposition coding \cite{cover1972broadcast}.
The work in \cite{bergmans1974simple}  then established the capacity region of the Gaussian BC with single-antenna
terminals.
For the  multiple-input multiple-output (MIMO) Gaussian BC,
the capacity region can be achieved by applying dirty paper coding (DPC) \cite{weingarten2006capacity}.
Moreover, the ergodic capacity and the outage capacity/probability of the fading BC
with perfect channel state information (CSI)
 at both the transmitter and receivers  were studied in \cite{li2001capacity}
 and \cite{li2001capacity_outage}, respectively.
 Compared to ergodic capacity, {the concept of outage assumes
 the transmission with a   predefined rate, which is more appropriate for
  applications with strict delay constraints}. Two types of outage probabilities were defined in \cite{li2001capacity_outage},
namely  the common outage probability (COP) and the
  individual outage probability (IOP).
  For the COP, an outage event occurs if any of the users are in outage.
  For the IOP, the outage events of individual users are considered.
   For the case where  CSI is not available  at the transmitter, the outage performance was analyzed
in \cite{zhang2009downlink}.

For the downlink MA scenario with $K$ users,
another key performance evaluation criterion is multiuser diversity, where  serving the user with the best instantaneous  channel gain
  yields  the optimal ergodic sum rate \cite{knopp1995information,tse1997optimal}.
However, user selection    requires a large amount of CSI feedback,
 which is difficult  to implement in practice.
Motivated by this, 
a significant amount of existing work is dedicated to harvesting  the multiuser diversity with only quantized
 CSI at the transmitter \cite{yu2006opportunistic,park2007scheduling}. One can
refer to the survey in \cite{love2008overview} for more details.
One of the most spectrally efficient approaches is to employ one-bit feedback for opportunistic user selection, which was proposed for the fading BC in
\cite{sanayei2005exploiting,sanayei2007opportunistic,nam2006opportunistic,somekh2007sum,
niu2010ergodic,wang2014exploiting}. 
The outage performance with one-bit feedback was  investigated in \cite{nam2006opportunistic,
niu2010ergodic}, and the use of one-bit feedback has also been applied  to the MIMO case in \cite{diaz2008asymptotic,xu2010two}.

This paper investigates  the block fading BC with one bit feedback
 from the new perspective of NOMA. The traditional  one-bit feedback schemes in \cite{sanayei2005exploiting,sanayei2007opportunistic,nam2006opportunistic,somekh2007sum,
niu2010ergodic,wang2014exploiting} opportunistically
  select  a single  user   for   transmission within each fading block, and hence do not
 achieve short-term fairness\footnote{In this paper, short-term fairness means that
 user fairness is  guaranteed within any fading block, whereas long-term fairness means that
  user fairness is guaranteed within a large number of fading blocks.} in general.
 Compared to these works, NOMA emphasizes short-term fairness,
   which is achieved by having  the base station
   transmit messages to all $K$ users simultaneously using superposition coding.
  { 
In comparison with the existing works on NOMA assuming availability of perfect CSI at the transmitter   (e.g.,
\cite{ding2014performance,timotheou2015fairness,
 liu2015cooperative,Ding2015cooperative_NOMA}), the proposed  NOMA scheme with one-bit feedback
  enjoys a lower overhead, especially when the number of users is large. It is worth pointing out
   that     this one-bit feedback scheme is aligned with how NOMA has been  implemented in practice. For example,   multiuser superposition transmission (MUST), a downlink two-user version of NOMA, has been
included  in 3rd generation partnership project long-term evolution advanced (3GPP-LTE-A) networks
\cite{LTE}. For MUST, the base station  needs to obtain partial CSI   to determine the ordering of the users,
and in \cite{LTE}, CSI feedback has   been particularly  highlighted  as
a potential enhancement  to assist the base station  in performing user ordering.}
   Most recently, in \cite{yang2016outage,shi2015outage}, the authors  have investigated the outage performance
   of NOMA with statistical CSI knowledge. However, the works in \cite{yang2016outage,shi2015outage} did not
   consider quantized CSI feedback and the proposed schemes are fundamentally different from our work.

 In this paper, {  a downlink  NOMA system with one-bit feedback is
investigated for
delay-sensitive applications. Therefore,  the outage probability is used as the relevant performance metric.}
Specifically, the COP is adopted as the  performance criterion,
 which is motivated by the fact that the COP captures the  event that outage occurs at any of the users and hence emphasizes short-term fairness compared to the IOP.
 We  derive a closed-form expression for the COP by first defining $(K+1)$ feedback events
 with respect to the number of channel gains exceeding a predefined threshold,
 and then analyzing the conditional COP for each event. The optimal diversity gains achieved by the considered NOMA scheme are derived
 under short-term and long-term power constraints, respectively.
 Our analysis shows  that the diversity gain
 under the long-term power constraint is twice as large as that under the short-term power constraint.

   Furthermore, in order to minimize the COP, we study a dynamic power allocation policy based on CSI feedback,
  i.e., different power allocation schemes are developed  for different feedback states.
      The formulated power allocation problem is challenging
 since the objective function for minimizing the COP is non-convex.
 To make this problem tractable, under the short-term power constraint,
  we first characterize  the properties of the  optimal power allocation solution, which can be used to  transform  the problem
 into a series of convex problems.
 Under the long-term power constraint,
  we apply a high signal-to-ratio (SNR) approximation and show that the approximated
   problem is convex.
 Our analysis shows that, for each  feedback event,
     the optimal solution is in the form of  two increasing
     geometric progressions.
     An efficient iterative search algorithm is proposed to determine the length of each
      geometric progression.
  {Numerical results reveal that one-bit feedback significantly  improves the outage
 performance of NOMA  compared to the case without CSI feedback.}


Throughout this paper, we use $\mathbb{P}(\cdot)$ to denote the probability of an event,
and  $\mathbb{E}(\cdot)$ denotes the expectation of a random variable.
In addition, $\{x_i\}$ denotes the sequence formed by all the possible $x_i$'s,
and $[1:K]$ denotes the set $\{1,\cdots,K\}$.
Furthermore, $\log(\cdot)$ denotes the logarithm that is taken to base 2; $\ln(\cdot)$ denotes
the natural logarithm; {$C_K^n\triangleq \frac{K!}{n!(K-n)!}$, for $n\leq K$}; and $[x]^+\triangleq\max\{x,0\}$.
{ Finally, ``$\doteq$''   denotes exponential equality, i.e.,
  $f(P)\doteq P^x$ implies  $\lim_{P\rightarrow \infty}\frac{\log f(P)}{\log P}=x$, and
   ``$\dot{\leq}$'' and ``$\dot{\geq}$'' are defined similarly.}

\section{System Model}\label{ii}
Consider a downlink NOMA scenario  with one single-antenna base station and $K$ single-antenna users.
{  Quasi-static block fading is assumed,
where the  channel gains from the base station to all users are  constant during one fading block consisting of $T$ channel uses, but change independently from one fading block to  the next fading block.}
The  base station sends $K$ messages to the users using the NOMA scheme, i.e., it sends
 $x(t,b)=\sum_{k=1}^K s_k(t,b)$ at time instant $t$ within fading block $b$, where $s_k(t,b)$ is
the transmitted signal (containing the information-bearing message and the power allocation coefficient)  for   user $k$ and the signals for different users are mutually independent.
 Accordingly,   user $i$ receives the following
\begin{align}\label{superposition}
  y_k(t,b)=h_k(b)\sum_{i=1}^K s_k(t,b)+n_k(t,b), \ t\in[1:T],
\end{align}
at  time instant $t$ within fading block $b$.  Here,
the noise samples $n_k(t,b)$ at   user $k$ are independent and identically distributed
complex Gaussian random variables  with zero mean and unit variance.
 {$h_k(b)$ denotes the channel gain from the base station to
 user $k$ in block $b$, which is  assumed
to be a zero mean circularly symmetric complex Gaussian random
variable with unit variance. Moreover,  the users have mutually independent
channel gains.}
This paper exclusively considers the case where all codewords span only a
single fading block, and the base station transmits one message to each user in each block  with
the same fixed rate $r_0$ bits per channel use (BPCU), in order to guarantee fairness \cite{timotheou2015fairness}.

For the sake of brevity, the fading block index $b$ will be omitted in the rest of this paper whenever
this does not cause any confusion. Assume that all users have perfect CSI and  compare their fading gains to  a predefined
    threshold, denoted by  $\alpha$. Particularly, given $h_k$, user $k$ feeds back in each fading block
    a single bit\footnote{   The one-bit feedback scheme  considered
    in this paper is the simplest form   of a quantized feedback scheme,
     and its overhead  is negligible when the length of each fading block is
moderate to large. However, this work can be viewed as a benchmark for future  studies of
 NOMA systems employing multiple-bit feedback.} ``$Q(h_k)$" to
    the base station  via
 a zero-delay reliable link, where $Q(h_k)=1$ if $|h_k|^2\geq \alpha$, and $Q(h_k)=0$, otherwise.


\subsection{User Ordering for NOMA}
 Denote the channel feedback sequence as $\{Q(h_k)\}\triangleq
 \{Q(h_1),\cdots,Q(h_K)\}$.
 Obviously, $\{Q(h_k)\}$ has $2^K$ possible realizations in each of which the elements are  0 or 1.
Based on these feedbacks,  the base station will perform power allocation for the $K$ users.  Thereby,
the base station focuses only on $(K+1)$ categories for the realizations of  $\{Q(h_k)\}$,
 and a corresponding random variable is  defined  in the following.
   \begin{Definition}\label{Defini_event}
    Define a random variable $N$ with respect to  the $K$-dimensional random binary feedback
    sequence $\{Q(h_k)\}$ as $N\triangleq K-\sum_{k=1}^K Q(h_k)$.
   Obviously, $N$ has $(K+1)$ possible realizations, and event $N=n$
    represents the case where $n$ users send ``0'' and the other
    $K-n$ users send ``1", $n\in[0:K]$.  \end{Definition}

  {  For event $N=n$, the base station uses three steps to determine the user ordering:
   (i) divide  the users into two groups
   corresponding to  feedbacks ``0" and
 ``1", denoted as $\mathcal{G}_{0|n}$ and $\mathcal{G}_{1|n}$, respectively;
 (ii) allocate the  ordering indices $\{1,\cdots,n\}$ to  the users  in $\mathcal{G}_{0|n}$,
 and  the  ordering  indices $\{n+1,\cdots,K\}$ to  the users in $\mathcal{G}_{1|n}$;
 (iii)  {\em randomly} index (order)  the users  in the same group  since the base station cannot
 distinguish their fading gains.}

Denote the channel gains for the ordered users  by $\{|h_{\pi_1}|^2,|h_{\pi_2}|^2, \cdots,|h_{\pi_K}|^2\}$,
 where $ \pi_k\in[1:K]$, and $\pi_i\neq \pi_j$ if $i\neq j$. Hence, for event $N=n$,
 { $Q(h_{\pi_k})=0$ if $1\leq k \leq n$}, and $Q(h_{\pi_k})=1$ if $n+1\leq k \leq K$.
{ Then, the base station broadcasts the superimposed message $\sum_{k=1}^K s_{\pi_k}(t)$ based on the power allocation policy  discussed in the next subsection, where $s_{\pi_k}(t)$
is the signal for user $\pi_k$ in the $t$-th channel use of a fading block.}
\begin{Remark}\label{Remark1}
  According to the applied user ordering principle, all channels $h_{\pi_k}$ are mutually independent
  if conditioned on event $N=n$. This is because the two groups  $\mathcal{G}_{0|n}$ and
  $\mathcal{G}_{1|n}$ are determined by event $N=n$, and all  users in the same group are randomly
  ordered.
\end{Remark}

{ \subsection{Successive Interference Cancellation (SIC)}
The users  employ SIC to decode their messages,
  based on the user ordering determined by the base station.
 As explained in the previous subsection,  the ordering of the channels is denoted as
 $\{|h_{\pi_1}|^2,|h_{\pi_2}|^2,
 \cdots,|h_{\pi_K}|^2\}$. In the SIC process,  user $\pi_k$ will sequentially  decode the messages
 of   users $\pi_l $, $l\in[1:k]$.
  Specifically,
   user $\pi_k$ will  successively  detect the message of  users $\pi_l$,  ${l}<{k}$,
   and then remove these messages from its  observation, such that
   the interference terms generated from user  ${\pi_1}$ to user ${\pi_l}$
   have been canceled when detecting
   the    message of  user $\pi_{l+1}$.}

 \subsection{Power Constraint}
 For any block,
   the power  allocated for user $\pi_k$, whose ordering index  in the SIC process is $k$,
is denoted as  $\mathcal{P}_k(\{Q(h_k)\})$.
While there are $2^K$ possible feedback sequences, the power allocation policy used at the base station
 will  depend only on which of the  $K+1$ events $N=n$ happens, i.e., the power allocation policy for all sequences corresponding to the same event are identical. Therefore, the power allocated to  user $\pi_k$ is denoted by  $P_{k,n}$, i.e.,  $\mathcal{P}_k(\{Q(h_k)\})=P_{k,n}$, for
   event $N=n$. 

We consider two different types of power constraints.
In particular, the {\em short-term} power constraint ensures that  the sum power of all users within any block is constrained. Specifically, the short-term power constraint
 requires that the total power allocated to  all   users within any block cannot exceed $P$, i.e.,
\begin{align}\label{power_short-term}
\sum_{k=1}^K{P_{k,n}}\leq P,\  \forall n\in[0:K].
\end{align}
In contrast, the considered {\em long-term} power constraint ensures that the average total transmission power is constrained, i.e., { 
\begin{align}\label{power_long-term2}
 \mathbb{E} \left[\sum_{k=1}^K \mathcal{P}_k(\{Q(h_k)\})\right]&=\sum_{k=1}^K \mathbb{E} \left[ \mathcal{P}_k(\{Q(h_k)\})\right]\leq P,
 \end{align}
 where the expectation of $\mathcal{P}_k(\{Q(h_k)\})$ can be calculated as
 \begin{align}
  \mathbb{E}\left[\mathcal{P}_k(\{Q(h_k)\})\right]&\stackrel{(a)}{=}\sum_{\mathbf{q}\in\mathcal{Q}}
  p(\mathbf{q})\mathcal{P}_k(\mathbf{q})\stackrel{(b)}{=}
  \sum_{n=0}^K \sum_{\mathbf{q}\in\mathcal{Q}_n}p(\mathbf{q})\mathcal{P}_k(\mathbf{q})\nonumber\\
  &\stackrel{(c)}{=}\sum_{n=0}^K \left[ P_{k,n}\sum_{\mathbf{q}\in\mathcal{Q}_n}p(\mathbf{q})\right]
  \stackrel{(d)}{=} \sum_{n=0}^K \left[P_{k,n}\mathbb{P}(N=n)\right]
\end{align}
where $(a)$ follows from the definitions $\mathcal{Q}\triangleq \left\{\mathbf{q}=
(q_1,\cdots,q_K):\ q_k\in\{0,1\},
  \forall k\in[1:K]\right\}$ and  $p(\mathbf{q})\triangleq\mathbb{P}(\{Q(h_k)\}=\mathbf{q})$;
  $(b)$ follows from the definition   $\mathcal{Q}_n\triangleq \left\{\mathbf{q}=(q_1,\cdots,q_K)\in\mathcal{Q}:
 K- \sum_{k=1}^K q_k=n\right\}$, $\forall n\in[0:K]$;
$(c)$ holds since $\mathcal{P}_k(\mathbf{q})=P_{k,n}$
if $\mathbf{q}\in\mathcal{Q}_n$ as shown at the beginning of this subsection; $(d)$ holds since
$\mathbb{P}(N=n)=\sum_{\mathbf{q}\in\mathcal{Q}_n}p(\mathbf{q})$ according to Definition \ref{Defini_event}.
Thus, the { long-term} power constraint in \eqref{power_long-term2} can be rewritten as
\begin{align}\label{power_long-term}
 \sum_{k=1}^K \mathbb{E} \left[ \mathcal{P}_k(\{Q(h_k)\})\right]= \sum_{k=1}^K\sum_{n=0}^K \left[ P_{k,n}\mathbb{P}(N=n)\right]
  =\sum_{n=0}^K \left[\mathbb{P}(N=n)\sum_{k=1}^K P_{k,n}\right]\leq P.
\end{align}
}


\begin{Remark}
 Both  types of power constraints are widely used in the related literature, e.g., \cite{sanayei2007opportunistic,somekh2007sum,kim2007expected,niu2010ergodic}. The short-term power constraint is appropriate  for applications with strict peak power constraints, whereas the
  long-term power constraint is appropriate for  applications with average power constraints.
\end{Remark}

 \section{Outage Probability}\label{iii}
  In this section, the outage probability of the NOMA system considered in Section \ref{ii} will
   be analyzed. However, first,  some useful  preliminary results are provided
   in the next subsection.
   \subsection{Preliminary Results}
    We first analyze of the conditional probability $ \mathbb{P}(|h_{\pi_k}|^2<x_k|N=n)$ for $x_k>0$, $k
  \in[1:K]$,
   where  random variable $N$ is defined in Definition \ref{Defini_event}.
   Based on the user ordering in Section \ref{ii}, we know that, for event $N=n$,
    $|h_{\pi_k}|^2<\alpha$ if $k\in[1:n]$, and   $|h_{\pi_k}|^2\geq\alpha$ otherwise.
    In addition, all channels $h_{\pi_k}$ are mutually independent
  if conditioned on event $N=n$, as explained in Remark \ref{Remark1}.
    Thus, we have
    \footnote{Note that,  when $n=0$ (i.e., event $N=0$), the probabilities  in  \eqref{CDF0} do
not exist; when $n=K$ (i.e.,  event $N=K$), the probabilities  in  \eqref{CDF1} do
not exist.
}
 \begin{align}
  \mathbb{P}(|h_{\pi_k}|^2&<x_k\ \big|\ N=n) {=}\mathbb{P}\left(|h_{\pi_k}|^2\leq x_k\;\big|\;|h_{\pi_k}|^2<\alpha\right)
 \nonumber\\& = \frac{ \mathbb{P}\left(|h_{\pi_k}|^2\leq x_k,|h_{\pi_k}|^2<\alpha\right)}{\mathbb{P}\left(|h_{\pi_k}|^2<\alpha\right)}\nonumber\\
 &=  \min\left\{\frac{1-e^{-x_k}}{1-e^{-\alpha}},1\right\},\  x_k\geq 0,k\in[1:n].  \label{CDF0}
  \end{align}
   Similarly, we have
  \begin{align}
     \mathbb{P}(|h_{\pi_k}|^2&<x_k\ \big|\ N=n) =  \left[1-e^{-(x_k-\alpha)}\right]^+,
     \  x_k\geq 0,k\in[n+1:K].
      \label{CDF1}
 \end{align}

Next,  the expressions for the signal-to-interference-plus-noise ratios (SINRs) at the receivers will be developed.
 As explained in Section \ref{ii}-B, SIC is adopted in the decoding process and
 the ordering of the channels is denoted as
 $\{|h_{\pi_1}|^2,|h_{\pi_2}|^2,
 \cdots,|h_{\pi_K}|^2\}$. Thus,
  the     SINR for  user $\pi_k$ to decode the message of user $\pi_l$  is given by \cite{cover1972broadcast}
\begin{align}\label{SNR}
  \textrm{SINR}_{l\rightarrow k}=
  \frac{P_{l,n}|{h}_{\pi_k}|^2}{|{h}_{\pi_k}|^2\sum_{m=l+1}^KP_{m,n}+1},
  \ l\in[1:k].
\end{align}
\subsection{Outage Probability}
 This paper adopts the COP \cite{li2001capacity_outage} as  performance criterion for the considered NOMA system { since short-term fairness  can be  guaranteed  with this criterion}.
  The COP  is provided
  in the following theorem.
\begin{theorem}\label{theorem_cop}
The COP of the considered one-bit NOMA scheme can be expressed as
\begin{align}
   \mathbb{P}^{\textrm{Common}}(\alpha,\{P_{k,n}\}) =\sum_{n=0}^K  \mathbb{P}_n(\alpha)\left[1-\prod_{k=1}^K (1-
\mathbb{P}_{k,n}^{\textrm{Indiv}}(\alpha,\mathbf{P}_n))\right],\label{common_outage_overall2}
\end{align}
where  $\mathbf{P}_n\triangleq\{P_{1,n},\cdots,P_{K,n}\}$  is the power allocation sequence for event $N=n$; $\mathbb{P}_n(\alpha)$ and $\mathbb{P}_{k,n}^{\textrm{Indiv}}(\alpha,\mathbf{P}_n)$ are defined as:  \begin{align}\label{Pr_n}
   \mathbb{P}_n(\alpha)\triangleq C_K^n(1-e^{-\alpha})^ne^{-\alpha(K-n)},
 \end{align}
 \begin{align}\label{P_i|n}
  \mathbb{P}_{k,n}^{\textrm{Indiv}}(\alpha,\mathbf{P}_n)\triangleq
  \left\{\begin{array}{ll}\min\left\{\frac{1-e^{-\hat{\zeta}_{k,n}}}{1-e^{-\alpha}},1\right\}, &
  k\in[1:n],\\
  \left[{1-e^{-(\hat{\zeta}_{k,n}-\alpha)}}\right]^+, &k\in[n+1:K],\end{array}\right.
\end{align}
with the definition $\hat{\zeta}_{k,n}\triangleq \max\{\zeta_{1,n},\cdots,\zeta_{k,n}\}$, and
\begin{align}\zeta_{k,n}=\frac{\hat{r}_0}{P_{k,n}-
\hat{r}_0\sum_{m=k+1}^K P_{m,n}}, \forall k\in[1:K], \textrm{ where }\hat{r}_0=2^{r_0}-1.\label{zeta}\end{align}
 \end{theorem}
 \begin{proof}
   Please refer to Appendix \ref{proof_theorem_cop}.
 \end{proof}


Note that in \eqref{zeta}, we have implicitly assumed that $\zeta_{k,n}\geq 0$, i.e.,
  \begin{align}\label{noma_constraint}
    {P_{k,n}\geq \hat{r}_0\sum_{m=k+1}^KP_{m,n}},\ \forall k\in[1:K-1],\ n\in[0:K].
  \end{align}
  Such a constraint on power allocation is typical for   NOMA systems \cite{ding2014performance,Ding2015cooperative_NOMA,timotheou2015fairness}, where a user
  with poorer channel conditions has to be allocated more power in order to guarantee fairness.
 In addition, in order to facilitate the use of different power constraints in the following discussions,
we express  $\{P_{k,n}\}$ as a function of  $\{\zeta_{k,n}\}$ as follows:
  \begin{align}\label{Power-zeta}
P_{k,n}=\frac{\hat{r}_0}{\zeta_{k,n}}+\hat{r}_0\sum_{m=k+1}^K
(\hat{r}_0+1)^{m-k-1}\frac{\hat{r}_0}{\zeta_{m,n}},\ \forall k\in[1:K],\ n\in[0:K],
\end{align}
which is obtained from   \eqref{zeta} by applying mathematical induction. {  Thus,
the sum power for event $N=n$ can be expressed as
\begin{align}\label{Power-zeta2}
  \sum_{k=1}^K P_{k,n}&=\sum_{k=1}^K \left(\frac{\hat{r}_0}{\zeta_{k,n}}+\hat{r}_0\sum_{m=k+1}^K
(\hat{r}_0+1)^{m-k-1}\frac{\hat{r}_0}{\zeta_{m,n}}\right)\nonumber\\
&=\sum_{k=1}^K \left(\frac{\hat{r}_0}{\zeta_{k,n}}+\frac{\hat{r}_0^2}{\zeta_{k,n}}\sum_{i=0}^{k-2}
(\hat{r}_0+1)^{i-2}\right)\nonumber\\
&=\sum_{k=1}^K\frac{(\hat{r}_0+1)^{k-1}\hat{r}_0}{\zeta_{k,n}}.
\end{align}}

\subsection{Diversity Gain}
In order to provide some insight into  the outage performance,
in this subsection, we analyze the diversity gains of the COP in \eqref{common_outage_overall2}
 under the short-term and long-term power constraints. The
 diversity gain is defined as follows.
 \begin{Definition}
 { The diversity gain based on the  COP is defined as
  \begin{align}\label{definition_diversity}
    d=-\lim_{P\rightarrow \infty} \frac{\log\mathbb{P}^{\textrm{Common}}}{\log P}.
  \end{align}}
     In addition, the diversity gain in \eqref{definition_diversity} can be also expressed  as
  $\mathbb{P}^{\textrm{Common}}\doteq P^{-d}$.
  \end{Definition}

Then, the following two lemmas provide the diversity gains of the COP
 under the short-term and long-term power constraints.
\begin{Lemma}\label{theorem_diversity_short}
 Under the short-term power constraint in \eqref{power_short-term},
  the maximum  achievable diversity gain
 of the considered NOMA scheme is $1$.
\end{Lemma}

\begin{proof}
We consider a specific power allocation scheme such that  the values of  the $\zeta_{k,n}$'s
 in \eqref{zeta} are  identical. Based on this power allocation scheme, we will show that
 a diversity gain of $1$ can be achieved. 
 The feedback threshold is set as $\alpha=\ln(2)$ for simplicity. Note that one can
 also choose any other value of $\alpha$ to achieve a diversity gain of $1$, which
 means that the maximum  diversity gain can be achieved for any $\alpha$.
   Then, a lower bound on the COP is derived to prove that a diversity gain of $1$
   is optimal for all possible power allocation schemes
 and all possible choices of  threshold  $\alpha$. Details of the proof are provided in Appendix \ref{proof_theorem_short}.
\end{proof}

\begin{Lemma}\label{theorem_diversity_long}
Under the long-term power constraint in \eqref{power_long-term},  the maximum achievable
diversity gain of the considered NOMA scheme  is $2$,
 which is achieved only if $\alpha$ satisfies   $\alpha\doteq P^{-1}$.
\end{Lemma}

\begin{proof}
We consider a specific power allocation scheme such that the $\zeta_{k,n}$'s
 in \eqref{zeta} have the same value for a given $n$.
 We also choose a threshold $\alpha$ such that  outages are not occurring for event $N=0$
 (i.e., all the users feed back ``1'').
   Then, a  lower bound on the COP is  derived to prove that a diversity gain of $2$
   is optimal for all  possible power allocation schemes
 and all possible choices of threshold  $\alpha$, under the long-term power constraint.
Details of the proof are provided in Appendix  \ref{proof_theorem_long}.
\end{proof}


\section{Power Allocation}\label{iv}
Existing works have demonstrated that power allocation has significant
impact on the outage performance in conventional  multiple access scenarios \cite{li2001capacity_outage,li2005outage,luo2005service}.
Motivated by this, in this section, we formulate a power allocation problem to minimize  the COP
 $ \mathbb{P}^{\textrm{Common}}$
in \eqref{common_outage_overall2}, under short-term and long-term power constraints.

\subsection{Problem Formulation }
The optimization problem for the short-term power constraint can be formulated as follows:
\begin{subequations}\label{problem_short}
\begin{align}
 &\min_{\alpha,\{P_{k,n}\}}\sum_{n=0}^K  \mathbb{P}_n(\alpha)
  \left[1-\prod_{k=1}^K (1-
\mathbb{P}_{k,n}^{\textrm{Indiv}}(\alpha,\mathbf{P}_n))\right] \\
  &\textrm{s.t. } \eqref{power_short-term} \textrm{ and } \eqref{noma_constraint},\
  P_{k,n}\geq 0,\  \forall k\in[1:K], n\in[0: K]. \label{constraint1}
\end{align}
\end{subequations}
Similarly, the optimization problem for the long-term power constraint can be formulated as follows:
\begin{subequations}\label{problem_long}
\begin{align}
 &\min_{\alpha,\{P_{k,n}\}}\sum_{n=0}^K  \mathbb{P}_n(\alpha)
  \left[1-\prod_{k=1}^K (1-
\mathbb{P}_{k,n}^{\textrm{Indiv}}(\alpha,\mathbf{P}_n))\right] \\
  &\textrm{s.t. } \eqref{power_long-term} \textrm{ and } \eqref{noma_constraint},\
  P_{k,n}\geq 0,\  \forall k\in[1:K], n\in[0: K].
\end{align}
\end{subequations}

To simplify the above two problems, variable transformation according to \eqref{zeta} is applied,
and the problem in \eqref{problem_short} is transformed into the following equivalent form:
\begin{subequations}\label{new_problem1}
\begin{align}
 \textrm{(P1) }&\min_{\alpha,\{\zeta_{k,n}\}}\sum_{n=0}^K  \mathbb{P}_n(\alpha)
 \left[1-\prod_{k=1}^K (1-
\mathbb{P}_{k,n}^{\textrm{Indiv}}(\alpha,\boldsymbol{\zeta}_n))\right] \label{obj2}\\
  \textrm{s.t. }&
  \sum_{k=1}^K\frac{(\hat{r}_0+1)^{k-1}\hat{r}_0}{\zeta_{k,n}}\leq P,\  n\in[0: K];\label{constraint2}\\
  &\zeta_{k,n}\geq 0,\ \forall k\in[1:K], n\in[0: K]\label{constraint22}.
\end{align}
\end{subequations}
where $\boldsymbol{\zeta}_n=\{\zeta_{1,n},\cdots,\zeta_{K,n}\}$ and $\mathbb{P}_{k,n}^{\textrm{Indiv}}$
becomes a function of $\boldsymbol{\zeta}_n$; { \eqref{constraint2} is based on
\eqref{power_short-term} and \eqref{Power-zeta2}.}
Note that, according to \eqref{Power-zeta},
 the optimal power allocation scheme can be found once the optimal values of $\{\zeta_{k,n}\}$ are obtained.
 Similarly, the problem in \eqref{problem_long} can be transformed into the following equivalent form:
\begin{subequations}\label{new_problem_long1}
\begin{align}
 \textrm{(P2) }&\min_{\alpha,\{\zeta_{k,n}\}}
 \sum_{n=0}^K  \mathbb{P}_n(\alpha)
 \left[1-\prod_{k=1}^K (1-
\mathbb{P}_{k,n}^{\textrm{Indiv}}(\alpha,\boldsymbol{\zeta}_n))\right]
 \label{obj_long1}\\
  \textrm{s.t. }&
  \sum_{n=0}^K\mathbb{P}(\alpha)
  \sum_{k=1}^K\frac{(\hat{r}_0+1)^{k-1}\hat{r}_0}{\zeta_{k,n}}\leq P;\label{constraint_long1}\\
  & \zeta_{k,n}\geq 0,\ \forall k\in[1:K], n\in[0: K].\label{constraint_long12}
\end{align}
\end{subequations}
The benefit of using  the transformed problems in \eqref{new_problem1} and \eqref{new_problem_long1} is that the number of constraints has been reduced. However, problems (P1) and (P2) still involve the non-convex objective function and  are difficult to  solve.
There are   $(K(K+1)+1)$ optimization variables in total, including $K(K+1)$ power variables $\zeta_{k,n}$
and one threshold variable $\alpha$.
In the subsequent subsections,  we first address the   power allocation problem
for a fixed threshold $\alpha$,
and then utilize a one-dimensional search  to find the optimal $\alpha$.

\subsection{Short-Term Power Constraint}\label{subsection_short}
 { 
 For a fixed $\alpha$,  $\mathbb{P}_n(\alpha)$ is also fixed, and therefore, the objective in \eqref{obj2}
is additive with respect to  subfunctions $\mathbb{P}_n(\alpha)
 \left[1-\prod_{k=1}^K (1-
\mathbb{P}_{k,n}^{\textrm{Indiv}}(\alpha,\boldsymbol{\zeta}_n))\right]$, where the $n$-th subfunction depends on
 variable vector $\boldsymbol{\zeta}_{n}$, $0\leq n\leq K$. Moreover, the constraints in \eqref{constraint2} and
 \eqref{constraint22}
are uncoupled with respect to the $(K+1)$ variable vectors $\boldsymbol{\zeta}_{n}$, $0\leq n\leq K$.
Hence,   the joint optimization problem (P1) can be decomposed  into $(K+1)$ decoupled   subproblems
{\em without loss of optimality},
where the $n$-th subproblem  has the following form:}
  \begin{subequations}\label{new_problem2}
\begin{align}
 &\max_{\boldsymbol{\zeta}_{n}}\  f_{1,n}(\alpha,\boldsymbol{\zeta}_{n})\triangleq  \prod_{k=1}^K (1-
\mathbb{P}_{k,n}^{\textrm{Indiv}}(\alpha,\boldsymbol{\zeta}_n)) \label{obj20.n}\\
  \textrm{s.t. }
  &\sum_{k=1}^K\frac{(\hat{r}_0+1)^{k-1}\hat{r}_0}{\zeta_{k,n}}\leq P,\label{constraint2.n}\
  \zeta_{k,n}\geq 0,\ \forall k\in[1:K].
\end{align}
\end{subequations}
As shown in \eqref{P_i|n}, $\mathbb{P}_{k,n}^{\textrm{Indiv}}$ is { a non-convex function.
The following  proposition shows how to simplify  $\mathbb{P}_{k,n}^{\textrm{Indiv}}$ }
for $k\in[n+1:K]$.
\begin{Proposition}\label{poposition_zeta1}
  The optimal solution of problem \eqref{new_problem2} satisfies
  $\hat{\zeta}_{k,n}\geq \alpha$, $\forall k\in[n+1:K]$, $n\in[0:K-1]$.
\end{Proposition}
\begin{proof}
{ From \eqref{P_i|n}, we have $\mathbb{P}_{k,n}^{\textrm{Indiv}}=0$ when $ \hat{\zeta}_{k,n}\leq \alpha$,
  $\forall k\in[n+1:K]$,   which means that, once   $ \hat{\zeta}_{k,n}\leq \alpha$,
further decreasing $\hat{\zeta}_{k,n}$ cannot decrease  $\mathbb{P}_{k,n}^{\textrm{Indiv}}$
nor increase $ f_{1,n}$ in \eqref{obj20.n}.
Thus, once   $ \hat{\zeta}_{k,n}\leq \alpha$,
we only need to consider the case of  $\hat{\zeta}_{k,n}=\alpha$, since this leads to a lower  power
consumption for user $k$ (i.e., $P_{k,n}$)
  than the case of $\hat{\zeta}_{k,n}<\alpha$, as is oblivious from \eqref{zeta}.
In summary, the case of $\hat{\zeta}_{k,n}<\alpha$ can be ignored and the optimal solution of the considered optimization problem satisfies $\hat{\zeta}_{k,n}\geq\alpha$.}
\end{proof}

{We can also simplify the functions $\mathbb{P}_{k,n}^{\textrm{Indiv}}$  for $k\in[1:n]$  by considering
$\hat{\zeta}_{k,n}\leq \alpha$ only} as explained in the following. As shown in \eqref{P_i|n}, if $\hat{\zeta}_{k,n}>\alpha$,
$\forall k\in[1:n]$, we have $\mathbb{P}_{k,n}^{\textrm{Indiv}}=1$,  and the
objective function in \eqref{obj20.n} has the worst value (i.e., $f_{1,n}=0$)  among the possible
values between 0 and 1. Exploiting  the above considerations, the problem in \eqref{new_problem2} can be simplified as follows:
\begin{subequations}\label{new_problem3}
\begin{align}
 \max_{\boldsymbol{\zeta}_{n}}\  &f_{1,n}(\alpha,\boldsymbol{\zeta}_{n})= \label{obj3.n}
 \prod_{k=1}^n \frac{e^{-\hat{\zeta}_{k,n}}-e^{-\alpha}}{1-e^{-\alpha}}
  \prod_{k=n+1}^K e^{-(\hat{\zeta}_{k,n}-\alpha)} \\
  \textrm{s.t. } &\eqref{constraint2.n}; \textrm{ and } \hat{\zeta}_{k,n}\leq \alpha,\forall k\in[1:n];\
  \hat{\zeta}_{k,n}\geq \alpha, \forall k\in[n+1:K]. \label{constraint3.n}
\end{align}
\end{subequations}
\begin{Remark}
The constraint in \eqref{constraint3.n} requires $P\geq \frac{({ \hat{r}_0+1})^n-1}{\alpha}$
to satisfy $\hat{\zeta}_{k,n}\leq \alpha$, $\forall k\in[1:n]$, as is oblivious from \eqref{Power-zeta}.
{ Note that if this requirement
  on the transmit power is not satisfied, i.e., $P<\frac{({ \hat{r}_0+1})^n-1}{\alpha}$,  $\mathbb{P}_{n}^{\textrm{Common}}=1-\prod_{k=1}^K (1-
\mathbb{P}_{k,n}^{\textrm{Indiv}})=1$ for    any power allocation,
 i.e., the COP for event $N=n$ must be 1 in this case.}
\end{Remark}

  To further simplify this problem, we introduce another  proposition which allows the elimination of $\hat{\zeta}_{k,n}$.
\begin{Proposition}\label{poposition_zeta}
  The optimal solution of problem \eqref{new_problem3} satisfies
  $\zeta_{k,n}\leq \zeta_{k+1,n}$, $\forall k\in[1:K-1]$.
\end{Proposition}
\begin{proof}
 { We first consider the case that $\zeta_{2,n}\leq\zeta_{1,n}$ for a fixed $\zeta_{1,n}$.
  { From the definition of  $\hat{\zeta}_{k,n}$ in Theorem \ref{theorem_cop} (i.e.,
  $\hat{\zeta}_{k,n}=\max\{\zeta_{1,n},\cdots,\zeta_{k,n}\}$),  we have $\mathbb{P}_{2,n}^{\textrm{Indiv}}=\mathbb{P}(|\tilde{h}_{2}|^2\leq \hat{\zeta}_{2,n})=\mathbb{P}(|\tilde{h}_{2}|^2\leq \max\{\zeta_{1,n},\zeta_{2,n}\})=\mathbb{P}(|\tilde{h}_{2}|^2\leq {\zeta}_{1,n})$  if $\zeta_{2,n}\leq\zeta_{1,n}$}, which means that, once $\zeta_{2,n}\leq\zeta_{1,n}$,
decreasing $\zeta_{2,n}$ cannot further decrease  $\mathbb{P}_{2,n}^{\textrm{Indiv}}$ nor increase
$f_{1,n}$ in \eqref{obj3.n}. In this case, we should set $\zeta_{2,n}=\zeta_{1,n}$,
   which requires less power for user $2$ (i.e., $P_{n,2}$)  
    compared to the choice  {$\zeta_{2,n}<\zeta_{1,n}$}, as is oblivious in \eqref{zeta}.
    Therefore, we can ignore the case $\zeta_{2,n}<\zeta_{1,n}$ and only consider the case $\zeta_{2,n}\geq\zeta_{1,n}$ without
  loss of optimality.
  Carrying out the above steps iteratively, the proposition is proved.} \footnote{Note that a similar proposition has been provided in \cite{timotheou2015fairness} to
solve a different optimization problem. However,  the  proof used here is different from the one in
\cite{timotheou2015fairness}.}
\end{proof}

Using Proposition \ref{poposition_zeta}, the problem in \eqref{new_problem3} can be transformed into
\begin{subequations}
\begin{align}
 \max_{\boldsymbol{\zeta}_{n}}\  &f_{2,n}(\alpha,\boldsymbol{\zeta}_{n})\triangleq
 \prod_{k=1}^n {\left(e^{-{\zeta}_{k,n}}-e^{-\alpha}\right)}
  \prod_{k=n+1}^K e^{-{\zeta}_{k,n}} \\
  \textrm{s.t. } &\eqref{constraint2.n}; \  {\zeta}_{k,n}\leq \alpha,\forall k\in[1:n];\
  {\zeta}_{k,n}\geq \alpha, \forall k\in[n+1:K];\\
  &\zeta_{k,n}\leq\zeta_{k+1,n},\ \forall k\in[1:K-1].
\end{align}
\end{subequations}

The objective function $f_{2,n}$ is still non-convex. However, by using   the natural logarithm of $f_{2,n}$,
 the problem in \eqref{new_problem3} (i.e., the $n$-th suboptimal problem of   problem (P1) in
\eqref{new_problem1} for a fixed $\alpha$) can be transformed into the following equivalent {\em convex} problem:
 \begin{subequations}
\begin{align}
  \textrm{(P1.$n$) } &\max_{
  \boldsymbol{\zeta}_n}\sum_{k=1}^n
  \ln\left({e^{-\zeta_{k,n}}-e^{-\alpha}}
  \right)-\sum_{k=n+1}^K {\zeta_{k,n}}\\
  &\textrm{s.t. } \sum_{k=1}^K\frac{(\hat{r}_0+1)^{k-1}\hat{r}_0}{\zeta_{k,n}}\leq P;\label{P1.n_C1}\\
   &\qquad\zeta_{1,n}\geq 0;\ \zeta_{n,n}\leq \alpha;\ \zeta_{n+1,n}\geq \alpha;\\
  &\qquad \zeta_{k,n}\leq\zeta_{k+1,n},\ \forall k\in[1:K-1].
\end{align}
\end{subequations}
One can calculate the  Hessian matrix of the objective function and the constraint in \eqref{P1.n_C1}
to verify that this problem is convex.
This convex optimization problem will be solved later in Section \ref{v} using corresponding numerical solvers, since a closed-form expression  for the optimal solution of problem (P1.$n$) is difficult to obtain.

Furthermore, the optimal value of $\alpha$ in   problem of (P1) in \eqref{new_problem1}
 can be found by applying  a one-dimensional  search.
It is worth pointing out that the optimal  $\alpha$ has a finite value.
{This is because the probability that all  users feed back the message ``0" goes to 1 (i.e.,
 $\mathbb{P}_K(\alpha)\rightarrow 1$)
  if $\alpha$ is sufficiently large,} which is equivalent to the case without
  CSI feedback. 


\subsection{Long-Term Power Constraint}\label{subsection_long}
\subsubsection{Approximation for High SNR}
Compared to problem  (P1),   problem  (P2) in \eqref{new_problem_long1} is more challenging,
since the decoupling approach used to solve   problem  (P1) is not applicable.
Here, in this subsection, we will focus on the high SNR approximation of
the objective function (i.e., $\mathbb{P}^{\textrm{Common}}$) in order  to simplify the problem.
Specifically, the objective function is first simplified for high SNR, the optimal solution of  this approximated problem is then obtained   for a fixed $\alpha$,
and finally { a  one-dimensional} search is used  to find the optimal value for $\alpha$.

Based on Propositions \ref{poposition_zeta1} and \ref{poposition_zeta},
 problem (P2) can be simplified as:
\begin{subequations}\label{new_problem_long2}
\begin{align}
 \textrm{(P3) }&\min_{\{\zeta_{k,n}\}}\sum_{n=0}^K  \mathbb{P}_n(\alpha)
      \label{obj_long2}
  f_{3,n}(\alpha,\boldsymbol{\zeta}_n)\\
  \textrm{s.t. }&\eqref{constraint_long1}\textrm{ and } \zeta_{k,n}\geq 0,k\in [1:n],n\in[1:K];
  \label{constraint_long20}\\& \zeta_{k,n}\geq \alpha,\
  \forall k\in[n+1:K],\ n\in[0:K-1];\label{constraint_long2}\\
  &\zeta_{k,n}\leq\zeta_{k+1,n},\ \forall k\in[1:K-1],n\in[0:K],\label{constraint_long22}
\end{align}
\end{subequations}
where $f_{3,n}(\alpha,\boldsymbol{\zeta}_n)\triangleq1-\prod_{k=1}^n \frac{\left[e^{-{\zeta}_{k,n}}-e^{-\alpha}\right]^+}
{1-e^{-\alpha}}\prod_{k=n+1}^K e^{-({\zeta}_{k,n}-\alpha)}$.
The following proposition shows that problem (P3)  can be approximately
transformed into a convex problem at high SNR.
\begin{Proposition}\label{poposition_P3}
At high SNR, problem (P3) in \eqref{new_problem_long2}  can be approximately
transformed into  convex problem (P4), which is defined as follows:
\begin{subequations}\label{P3}
\begin{align}
\textrm{(P4) }&\min_{\{\zeta_{k,n}\}} \sum_{n=0}^K\mathbb{P}_n(\alpha)\left[\sum_{k=1}^n
\frac{\zeta_{k,n}}{1-e^{-\alpha}}+
\sum_{k=n+1}^K(\zeta_{k,n}-\alpha)\right]\label{P3_obj}\\
  \textrm{s.t. }&\eqref{constraint_long1}\textrm{ and } \zeta_{k,n}\geq 0,k\in [1:n],n\in[1:K];\label{P3_C2}\\
     &\zeta_{k,n}\geq \alpha,k\in [n+1:K],\ n\in[0:K-1];\label{P3_C3}\\
  &  \zeta_{k,n}\leq\zeta_{k+1,n},\ k\in[1:K-1],n\in[0:K].\label{P3_C4}
\end{align}
\end{subequations}
\end{Proposition}
\begin{proof}
 Please refer to Appendix \ref{proof_proposition_P3}.
\end{proof}
\begin{Remark}
  Although the approximation in Proposition \ref{poposition_P3} is obtained for  high SNR,
  even in the moderate SNR regime, the resulting suboptimal solution
    can still provide a significant performance gain
  compared to  benchmark schemes, as shown later in Section \ref{v},
\end{Remark}
\subsubsection{Optimal Solution of Problem (P4)}
Problem (P4) is a convex optimization problem for a given $\alpha$.
 To further simplify this problem, we define a new problem as follows.
\begin{Definition}
A new convex optimization problem, denoted by (P5), is obtained by removing
the last constraint in \eqref{P3_C4} of problem (P4).
 \end{Definition}

 We will show  in Proposition \ref{poposition_zeta2} that  problems (P4) and (P5) are exactly equivalent, i.e.,
 the optimal solution of problem (P5) automatically satisfies  constraint \eqref{P3_C4}.
  The Lagrangian function of the optimal solution for  problem (P5) is given by
\begin{align}
 & \mathcal{L}(\{\zeta_{k,n}\},w,\{\lambda_{k,n}\})\triangleq\mathbb{P}_n(\alpha)\left[\sum_{k=1}^n
\frac{\zeta_{k,n}}{1-e^{-\alpha}}+
\sum_{k=n+1}^K(\zeta_{k,n}-\alpha)\right] \nonumber\\
&+\omega\left(\sum_{n=0}^{K}\mathbb{P}_n(\alpha) \sum_{k=1}^K\frac{(\hat{r}_0+1)^{k-1}\hat{r}_0}{\zeta_{k,n}}
  -{P}\right)\!-\!\sum_{n=1}^K \sum_{k=1}^n \lambda_{k,n}\zeta_{k,n}\!-\!
  \sum_{n=0}^{K-1}\sum_{k=n+1}^K \lambda_{k,n}(\zeta_{k,n}\!-\!\alpha),
\end{align}
where $\lambda_{k,n},\omega\geq 0$ are Lagrange multipliers. { The
 Karush-Kuhn
Tucker (KKT) conditions are given by}
\begin{align}\label{partialL}
  \frac{\partial \mathcal{L}}{\partial \zeta_{k,n}}\!=\!\left\{\begin{array}{ll}
   \frac{\mathbb{P}(\alpha) }{1-e^{-\alpha}}\!-\!\frac{\omega\mathbb{P}(\alpha) (\hat{r}_0+1)^{k-1}
   \hat{r}_0}{\zeta_{k,n}^2}\!-\!\lambda_{k,n}=0, &\textrm{if }k\in[1:n],n\in[1:K];\\
    {\mathbb{P}(\alpha) }-\frac{\omega\mathbb{P}(\alpha) (\hat{r}_0+1)^{k-1}
   \hat{r}_0}{\zeta_{k,n}^2}\!-\!\lambda_{k,n}=0, &\textrm{if }k\in[n\!+\!1:K],n\in[0:K\!-\!1].
  \end{array}\right.
\end{align}
The complementary slackness conditions can be expressed as follows:
\begin{subequations}
\begin{align}
&\omega\left(\sum_{n=0}^{K}\mathbb{P}_n(\alpha)
\sum_{k=1}^K\frac{(\hat{r}_0+1)^{k-1}\hat{r}_0}{\zeta_{k,n}}-P\right)=0\label{kkt0}\\
 & \lambda_{k,n}\zeta_{k,n}=0\textrm{ if }k\in[1:n],n\in[1:K];\label{kkt1}\\
 & \lambda_{k,n}(\zeta_{k,n}\!-\!\alpha)=0 \textrm{ if }k\in[n+1:K],n\in[0:K-1].  \label{kkt2}
\end{align}
\end{subequations}
From \eqref{partialL} and \eqref{kkt0}-\eqref{kkt2}, we have $\omega>0$,
 $\lambda_{k,n}=0$, for $k\in[1:n],\ n\in[1:K]$,
and the optimal $\zeta_{k,n}$ can be expressed as follows:
\begin{align}\label{P4_zeta}
  \zeta_{k,n}=\left\{\begin{array}{ll}
  \sqrt{\omega (\hat{r}_0+1)^{k-1}\hat{r}_0(1-e^{-\alpha})},
  &\textrm{if }k\in[1:n],n\in[1:K];\\
 \sqrt{\frac{\omega \mathbb{P}(\alpha)(\hat{r}_0+1)^{k-1}\hat{r}_0}{\mathbb{P}(\alpha)-\lambda_{k,n}}}  , &\textrm{if }k\in[n+1:K],n\in[0:K-1].
  \end{array}\right.
  \end{align}
 The Lagrange multipliers are difficult to obtain directly.
Hence, we first  study the  properties
of the optimal power allocation. The following proposition
 will demonstrate that the constraint in \eqref{P3_C4} is always satisfied.
 \begin{Proposition}\label{poposition_zeta2}
   The optimal solution  of problem (P5) in \eqref{P4_zeta} satisfies
   $\zeta_{k,n}\leq\zeta_{k+1,n}$, $\forall k\in[1:K-1]$, $n\in[0:K]$, i.e.,
   problems (P4) and (P5) are equivalent.
 \end{Proposition}
\begin{proof}
  Please refer to Appendix \ref{proof_theorem_longpower}.
\end{proof}

   By using  this proposition and also constraint \eqref{P3_C3}, one can observe that, if $\zeta_{k,n}=\alpha$
   for a given $k\in[n+1:K]$ and $n\in[0:K-1]$, $\zeta_{l,n}=\alpha$ also holds $\forall l\in[n+1:k]$.
   Hence, we can define a series of integers representing the number of         $\zeta_{k,n}$'s that are equal to $\alpha$  as follows.
   \begin{Definition}\label{definition_in}
   For each $n$, denote  $i_n\in[0:K-n]$ as the number of $\zeta_{k,n}$'s whose values are equal to $\alpha$, i.e.,
   $\zeta_{k,n}=\alpha$, for $k\in[n+1:n+i_n]$ and $\zeta_{k,n}>\alpha$  for $k\in[n+i_n+1:K]$.
   \end{Definition}

   Once all $i_n$'s are given, the optimal solution of the $\zeta_{k,n}$'s
can be easily obtained as follows.
\begin{theorem}\label{theorem_power_long}
 If   all integers $i_n\in[0:K-n]$ defined in Definition \ref{definition_in} are known,
 the optimal solution of problems (P4) and (P5) can be expressed as follows:
    \begin{align}\label{eq_theorem_long}
  \zeta_{k,n}=\left\{\begin{array}{ll}
  \sqrt{\omega (\hat{r}_0+1)^{k-1}\hat{r}_0(1-e^{-\alpha})},
  &\textrm{if }k\in[1:n],\\
  \qquad\alpha,&\textrm{if }k\in[n+1:n+i_n],\\
 \sqrt{{\omega (\hat{r}_0+1)^{k-1}\hat{r}_0}}  , &\textrm{if }k\in[n+i_n+1:K],
  \end{array}\right.
  \end{align}
 for each $n\in[0:K]$, where \begin{align}\label{omega}
\sqrt{w}=\frac{\sum_{n=0}^K\mathbb{P}_n(\alpha)A_n(i_n)}
  {P-\sum_{n=0}^K\mathbb{P}_n(\alpha)B_n (i_n)},\end{align}
  and \begin{align}
     A_n(i_n)&\triangleq \sum_{k=1}^n
   \sqrt{\frac{{(\hat{r}_0+1)^{k-1}\hat{r}_0}}{{1-e^{-\alpha}}}}
   +\sum_{k=n+i_n+1}^K \sqrt{(\hat{r}_0+1)^{k-1}\hat{r}_0},\label{An}\\
  B_n(i_n)&\triangleq  \sum_{k=n+1}^{n+i_n}\frac{{(\hat{r}_0+1)^{k-1}\hat{r}_0}}{\alpha}.
  \label{Bn}
  \end{align}
  Note that $A_0(i_0)\triangleq 0$ if $i_0=K-n$, and $B_n(i_n)\triangleq0$ if $i_n=0$, $\forall n\in[0:K]$.
 \end{theorem}
\begin{proof}
Since $\zeta_{k,n}>\alpha$ if $k\in[n+i_n+1:K]$ as shown in Definition \ref{definition_in},
 { we have $\lambda_{k,n}=0$ for $k\in[n+i_n+1:K]$ as shown in \eqref{kkt2}}.
 Hence, from \eqref{P4_zeta}, the expression for $\zeta_{k,n}$ in \eqref{eq_theorem_long}
   can be obtained. Moreover, since $\omega>0$ in \eqref{kkt0}, we have
   \begin{align}\label{omega_P}\sum_{n=0}^{K}\mathbb{P}_n(\alpha)
\sum_{k=1}^K\frac{(\hat{r}_0+1)^{k-1}\hat{r}_0}{\zeta_{k,n}}=P.\end{align}
Substituting {the $\zeta_{k,n}$ in  \eqref{eq_theorem_long}}
   into the above equality, we obtain $\omega$ as shown in \eqref{omega}.
\end{proof}

\begin{Remark}
  Theorem \ref{theorem_power_long} shows that the optimal solution of
  $\{\zeta_{1,n},\cdots,\zeta_{K,n}\}$ is in the form of two increasing  geometric progressions
  and some constant $\alpha$ between them. Interestingly, parameter $n$ which represents    the feedback event $N=n$ only affects
  the lengths of the two  geometric progressions, but does not affect the value of the elements.
\end{Remark}

\subsubsection{Search Algorithm for $\{i_n^*\}$}
  The work left is to determine the unique integer sequence, denoted by $\{i_n^*\}$, such that
   all complementary slackness conditions are satisfied. We know that $\lambda_{k,n}=0$ for $k\in[1:n]$,
    so we only need to choose   $\{i_n^*\}$ such that
  \begin{align} \label{kkt_constraint} \textrm{$\lambda_{k,n}\geq 0$
    for $k\in[n+1:n+i_n^*]$ and $\zeta_{k,n}>\alpha$  for $k\in[n+i_n^*+1:K]$.}\end{align}
   Note that, given $\{i_n\}$,   since $\zeta_{k,n}^{(t)}=\alpha$ for $k\in[n+1:n+i_n]$ in \eqref{P4_zeta},
    $\lambda_{k,n}$ can be obtained  as
    \begin{align}\label{lambda}\lambda_{k,n}=\mathbb{P}_n(\alpha)
   \left(1-\frac{\omega(\hat{r}_0+1)^{k-1}\hat{r}_0}
  {\alpha^2}\right), \ k\in[n+1:n+i_n].\end{align}

  Unfortunately, a closed-form solution for the $i_n^*$ does not exist. Hence, we design an
  efficient  iterative algorithm to find $\{i_n^*\}$, as summarized in Algorithm I.
 Specifically, the search starts from $i_n^{(1)}=0$, $\forall n\in[0:K]$, and the main idea is to
  narrow down the search range of a certain number of $i_n^*$'s in  each iteration,
  by enlarging the lower bounds on these $i_n^*$'s.

 { The following theorem  ensures that
 the unique sequence $\{i_n^*\}$ can be found by the proposed algorithm, i.e., Algorithm I converges.}
 \begin{theorem}\label{theorem_algorithm}
The strategy proposed in Algorithm I, updating each
  $i_n^{(t)}$ satisfying   $\zeta_{n+i_n^{(t)}+1,n}^{(t)}\leq\alpha$ as
 $i_n^{(t+1)}=\arg\max_{i\in\left[i_n^{(t)}:K-n\right]}
  \{i:\zeta_{n+i}^{(t)}\leq \alpha\}$, guarantees that $\{i_n^*\}$ must be found.
 \end{theorem}
   \begin{proof}
   Please  refer to Appendix \ref{proof_theorem_algorithm}.
   \end{proof}


  \begin{table}
\hrule
\vspace{1mm}
\noindent {\bf Algorithm I}: Proposed search for $\{i_n\}$ defined in Definition \ref{definition_in}. \label{Table: Table I}
\vspace{0mm}
\hrule

\begin{enumerate}
   \item Initialize $t=1$, $i_n^{(1)}=0$ for $ n\in[0:K]$, and $\lambda_{k,n}^{(1)}=0$ for $k\in[n+1:K]$, $n\in[0:K]$.
  \item {The $t$-th iteration:}
 \begin{enumerate} 
  \item Update $\omega^{(t)}$, $\lambda_{k,n}^{(t)}$, and $\zeta_{k,n}^{(t)}$  
     in  \eqref{omega}, \eqref{lambda}, and \eqref{eq_theorem_long}, respectively.
   \item If $i_n^{(t)}=K-n$ or $\zeta_{n+i_n^{(t)}+1,n}^{(t)}>\alpha$, $\forall n\in[0:K]$, break the loop and the algorithm ends.
   \item { Else, for each $n$ satisfying  $\zeta_{n+i_n^{(t)}+1,n}^{(t)}\leq\alpha$, set $i_n^{(t+1)}$ as $$i_n^{(t+1)}=\arg\max_{i\in\left[i_n^{(t)}+1:K-n\right]}
   \{i:\zeta_{n+i,n}^{(t)}\leq \alpha\},$$ whereas,
   for each $n$ satisfying  $\zeta_{n+i_n^{(t)}+1}^{(t)}>\alpha$,  set $i_n^{(t+1)}$ as $
   i_n^{(t+1)}=i_n^{(t)}$.}

   \end{enumerate}
   \item Update $t=t+1$ and repeat Step 2) until $\{i_n^*\}$ is found.
\end{enumerate}
\hrule
\end{table}

    According to \eqref{eq_theorem_long}, $\zeta_{k,0}^{(t)}=\zeta_{k,n}^{(t)}$, $\forall k\in[n+1]$, $n\in[0:K-1]$.
    Thus, according to Step 2-c in Algorithm I, we obtain  $i_n^{(t)}=i_0^{(t)}-n$ if $ n\in[0:i_0^{(t)}-1]$
    and $i_n^{(t)}=0$, otherwise. {Since $i_0^{(t)}\in[0:K]$,  at most $K+1$ iterations  are required
      to find $\{i_n^*\}$}, which means that the proposed algorithm  enjoys  low complexity compared to
      an exhaustive search which would have complexity $O((K+1)!)$.

\section{Numerical Results}\label{v}
In this section, computer simulation  results are provided to evaluate the outage performance of the considered NOMA scheme
with one-bit feedback.

\subsection{Benchmark Schemes}\label{subsection_benchmark}
Some   benchmark transmission and power allocation schemes  are considered as explained in the following.
\subsubsection{TDMA  Scheme}
{The first benchmark scheme is  TDMA transmission with one-bit feedback since it is
equivalent to any orthogonal multiple access scheme \cite[Sec. 6.1.3]{tse2005fundamentals}.}
For TDMA transmission, assume that each fading block is equally divided into $K$ time slots, and   user
$k$ is served during the $k$-th time slot. The power allocated to   user $k$ is denoted by $P_{k,n}^T$
for each event $N=n$, where $N=n$ is defined in Definition \ref{Defini_event} based on the feedback sequence.
 The short-term and long-term power constraints in TDMA are $\frac{1}{K}\sum_{k=1}^KP_{k,n}^T\leq P$
and $\frac{1}{K}\sum_{n=0}^K\mathbb{P}(N=n)\sum_{k=1}^KP_{k,n}^T\leq P$, respectively. Furthermore, redefine  $\{\zeta_{k,n}\}$
in \eqref{zeta} as $\zeta_{k,n}=\frac{2^{Kr_0}-1}{P}$. 
The  short-term and long-term power constraints can be rewritten as follows:
\begin{align}
 & \frac{2^{Kr_0}-1}{K}\sum_{k=1}^K\frac{1}{\zeta_{k,n}}\leq P,\ \forall n\in[0:K]; \label{TDMA_short}\\
  & \frac{2^{Kr_0}-1}{K} \sum_{n=0}^K\mathbb{P}_n(\alpha)\sum_{k=1}^K\frac{1}{\zeta_{k,n}}\leq {P} ,
  \label{TDMA_long}
\end{align}
respectively. Now, similar to  problems (P1) an (P2) in \eqref{new_problem1} and \eqref{new_problem_long1},
 one can formulate two power allocation problems for TDMA transmission under  short-term and long-term power constraints as shown  in \eqref{TDMA_short} and \eqref{TDMA_long}, respectively.
We can solve the two  new problems using similar approaches as in Section \ref{iv}. The details are omitted here
due to space limitations.
\subsubsection{Fixed NOMA} In order to show the benefits of the proposed power allocation schemes,
  NOMA with fixed power allocation using one-bit feedback is used as the second benchmark scheme.
   {Due to its simplicity, fixed NOMA has been adopted in many relevant works (e.g.,
   \cite{ding2014performance,liu2015cooperative}).}
   Specifically, we also utilize the NOMA
transmission scheme  in Section \ref{ii}, but fix the power allocation as follows:
under the short-term power constraint, we let
$\zeta_{k,n}=\frac{(\hat{r}_0+1)^K-1}{P}$, $\forall k\in[1:K]$, $n\in[0:K]$;
under the long-term power constraint we let $\zeta_{k,n}=\frac{[(\hat{r}_0+1)^K-1](K+1)
\mathbb{P}_n(\alpha)}{P}$, $\forall k\in[1:K]$, $n\in[0:K]$. Note that such  power allocation
schemes have  been utilized in Appendices \ref{proof_theorem_short} and \ref{proof_theorem_long}
to prove Lemmas \ref{theorem_diversity_short} and \ref{theorem_power_long}, respectively.
The optimal $\alpha$ is also obtained via a one-dimensional search, for fairness of comparison.

\subsubsection{NOMA without Feedback} In order to show the benefits of using one-bit feedback,
 the third benchmark scheme is NOMA without CSI feedback, i.e., { the base station only has the average CSI information, but does not have the instantaneous CSI nor the  ordering information  \cite{yang2016outage}}. In this case, the base station  randomly orders the users;
 the long-term power constraint reduces to the short-term power constraint and utilizes
  only one power allocation within each fading block.
 Note that NOMA without CSI is a special case of the considered
     NOMA with one-bit feedback when we set $\alpha=0$ or $\alpha=\infty$.
\subsubsection{NOMA with Perfect CSI} Finally, NOMA with perfect CSI
is considered as a lower bound on the COP.
With perfect CSI, the base station informs the users  of
the optimal ordering of all channel gains, and knows the required power threshold for the users within any block
for correct decoding. In this case, {  we only consider the short-term power constraint,
where an outage event occurs if the required power threshold is larger than $P$ \cite{caire1999optimum}}.
For the long-term power constraint, an outage probability of zero can be achieved when $P$ is
sufficiently large, as shown in \cite{li2001capacity_outage}, which will not be considered in this section.
\subsection{Short-Term Power Constraint}
{\renewcommand\baselinestretch{1}\selectfont
\begin{figure}[tbp]
    \hspace{0.3cm}\begin{minipage}[t]{0.45\linewidth} 
    \centering
    \includegraphics[width=3.3in]{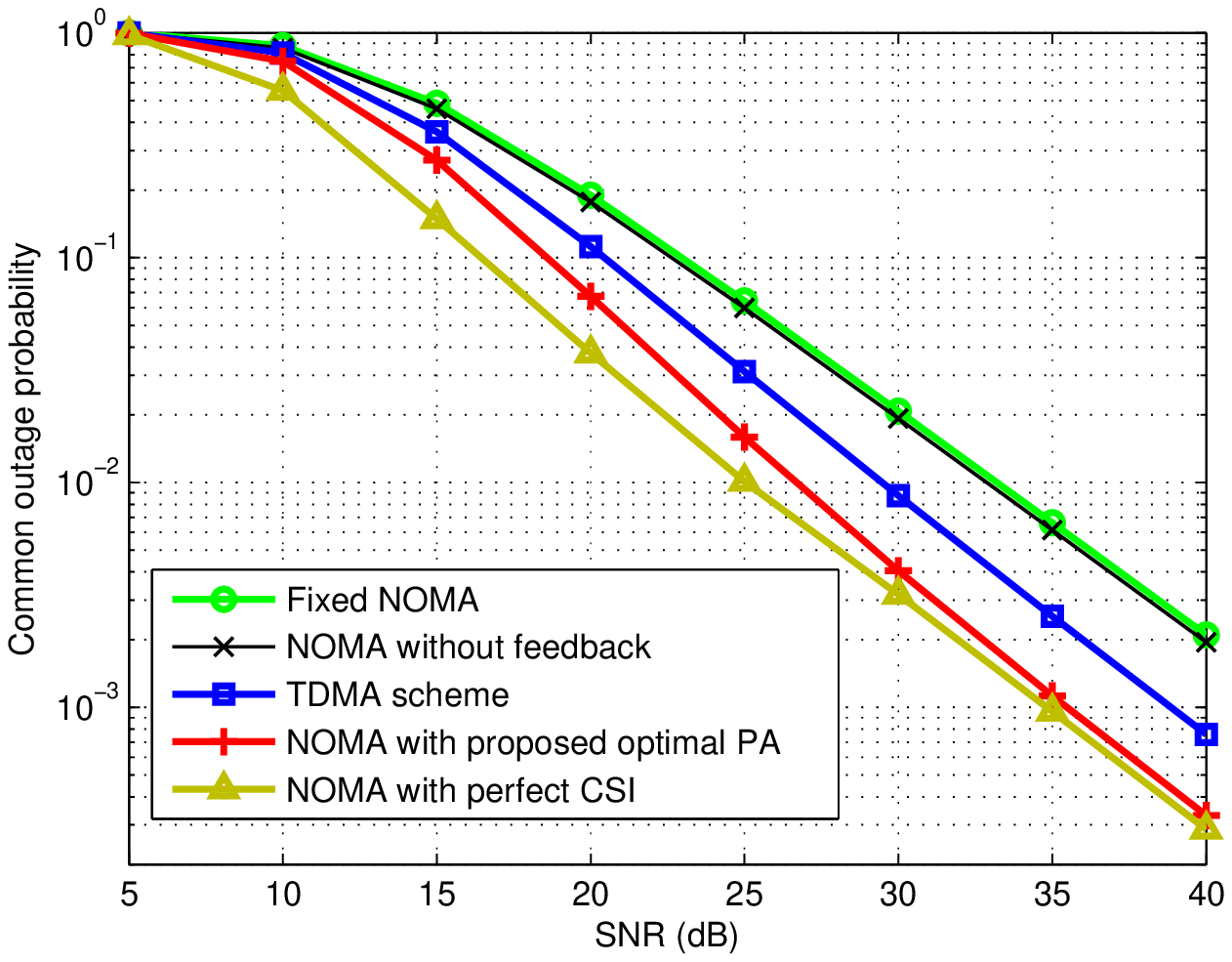}\vspace{-1em}
    \caption{COP versus SNR under the short-term power constraint, where  $K=3$, the target rate is $r_0=1$ BPCU
    for each user, and ``PA'' stands for  ``power allocation''.}
    \label{fig:side:a1}
  \end{minipage}%
  \hspace{0.8cm}
  \begin{minipage}[t]{0.45\linewidth}
    \centering
    \includegraphics[width=3.3in]{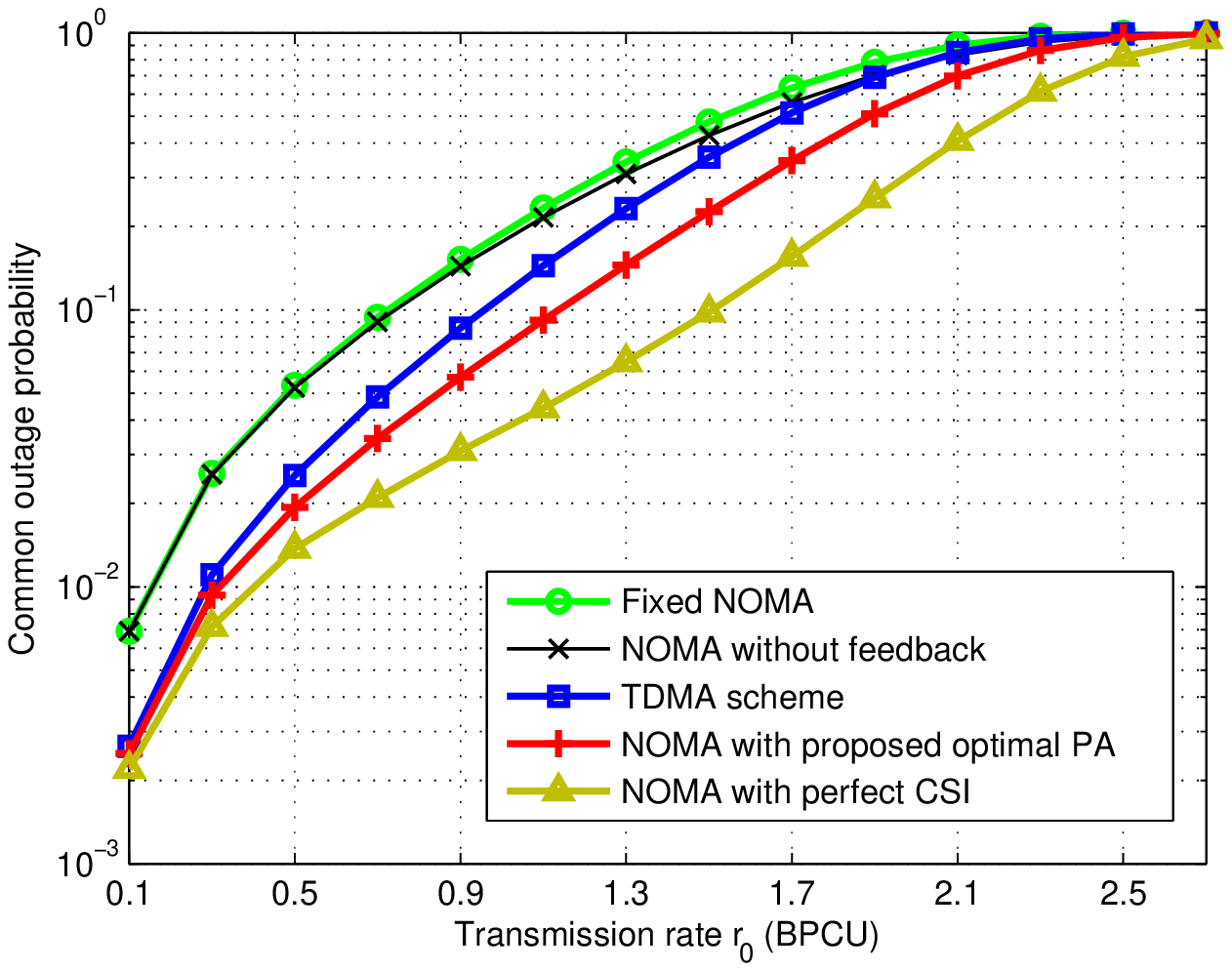}\vspace{-1em}
    \caption{COP versus transmission rate under the short-term power constraint, where  $K=3$, and the SNR is $20$ dB.}
    \label{fig:side:a2}
  \end{minipage}
\end{figure}
\par}

This subsection focuses on the outage performance of NOMA with one-bit feedback
under the short-term power constraint in \eqref{problem_short}.
Figs. \ref{fig:side:a1}, \ref{fig:side:a2}, and \ref{fig:side:b1} compare the outage performance of
 NOMA employing  the   optimal power allocation scheme proposed in Section \ref{subsection_short} with
  the benchmark schemes
defined in the previous subsection as a function of the SNR,   the transmission rate $r_0$, and the  number
of users $K$,
respectively. These figures demonstrate  that NOMA with optimal power allocation outperforms
the TDMA scheme, fixed NOMA, and NOMA without feedback.  As can be observed
in Fig. \ref{fig:side:a1}, all the curves
have almost the same slope at high SNR, but a constant gap exists between the proposed scheme
and each benchmark scheme. This is because all the schemes achieve the same diversity gain of $1$
(Lemma \ref{theorem_diversity_short})
under the short-term power constraint.
In addition, the performance of the proposed NOMA scheme with one-bit feedback  approaches that of NOMA with perfect CSI at high SNR, which means that the one-bit feedback is  effectively used
 by the proposed scheme to improve the outage performance.
Fig. \ref{fig:side:a2} reveals that  NOMA with the proposed optimal power allocation
has almost the same COP as the TDMA scheme when $r_0=0.1$, but outperforms the latter
as $r_0$ increases. For example, when $r_0=1.3$, these two schemes have  COPs of
approximately  $0.15$ and $0.23$, respectively. Finally, as shown in Fig. \ref{fig:side:b1},
the COPs of all  schemes increase with the number of the users. Particularly,
the gap between the proposed NOMA scheme and the TDMA scheme is enlarged as $K$ increases.
{This is because, compared to the orthogonal TDMA scheme,
NOMA is more spectrally efficient in the sense that all users are served  simultaneously.}

{\renewcommand\baselinestretch{1}\selectfont
\begin{figure}[tbp]
    \hspace{0.3cm}\begin{minipage}[t]{0.45\linewidth} 
    \centering
    \includegraphics[width=3.3in]{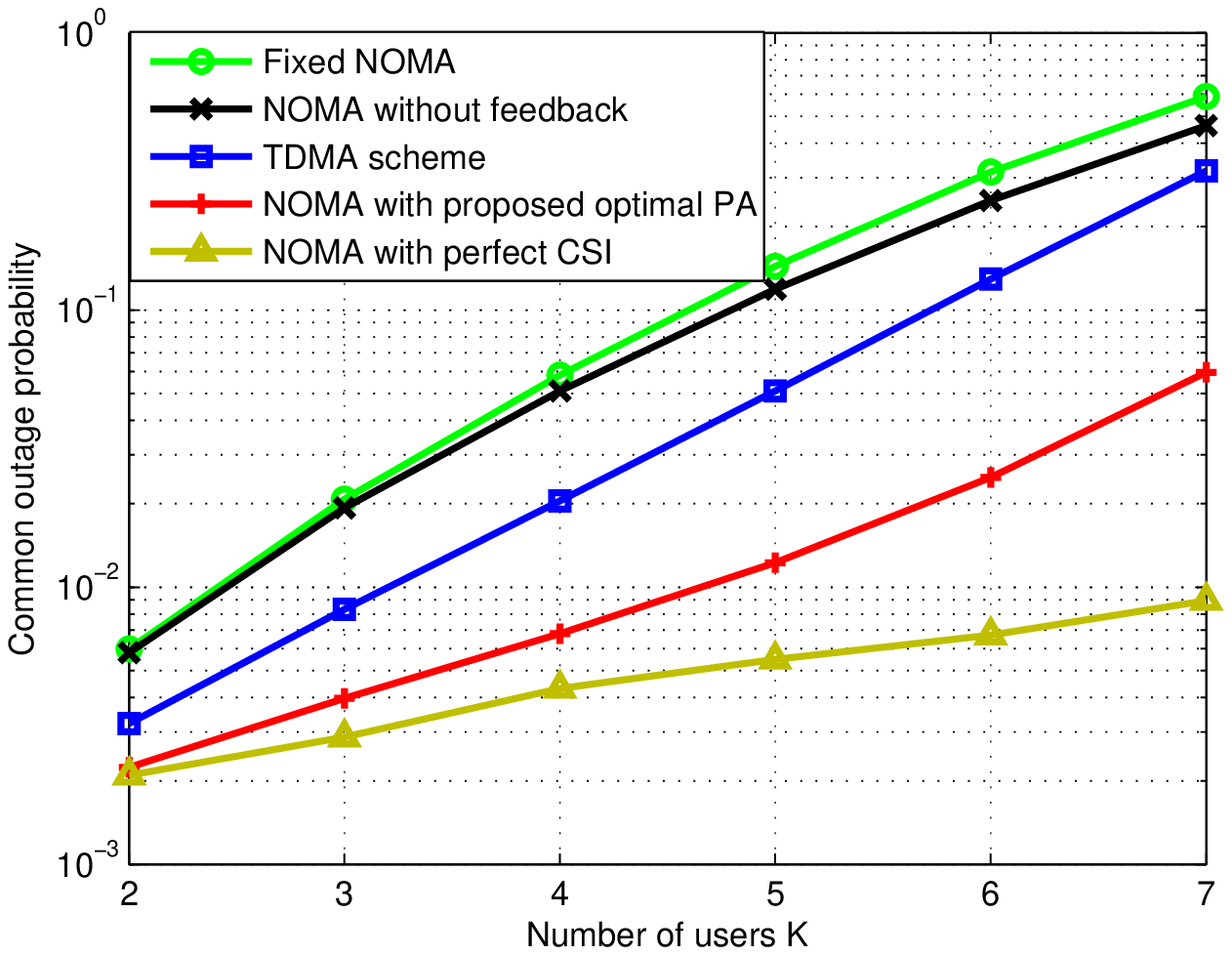}\vspace{-1em}
    \caption{COP versus the number of users under the short-term power constraint, where  the target transmission rate is $r_0=1$ BPCU
    for each user, and  the SNR is $30$ dB.}
    \label{fig:side:b1}
  \end{minipage}%
  \hspace{0.8cm}
  \begin{minipage}[t]{0.45\linewidth}
    \centering
    \includegraphics[width=3.3in]{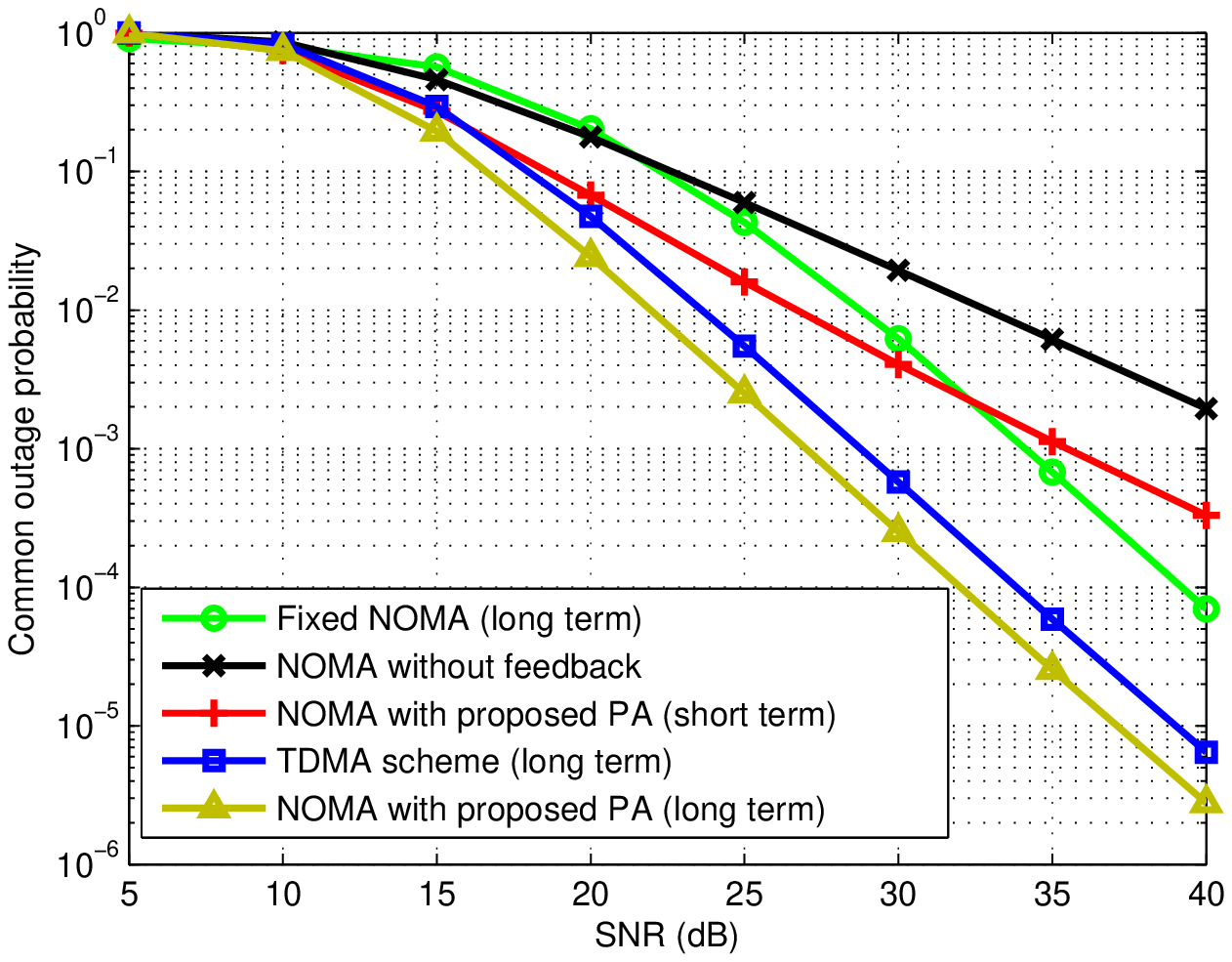}\vspace{-1em}
    \caption{COP versus SNR under the long-term power constraint, where the number of users is $K=3$, and the target transmission rate is $r_0=1$ BPCU
    for each user.}
    \label{fig:side:b2}
  \end{minipage}
\end{figure}
\par}
\subsection{Long-Term Power Constraint}
{\renewcommand\baselinestretch{1}\selectfont
\begin{figure}[tbp]
    \hspace{0.3cm}\begin{minipage}[t]{0.45\linewidth} 
    \centering
    \includegraphics[width=3.3in]{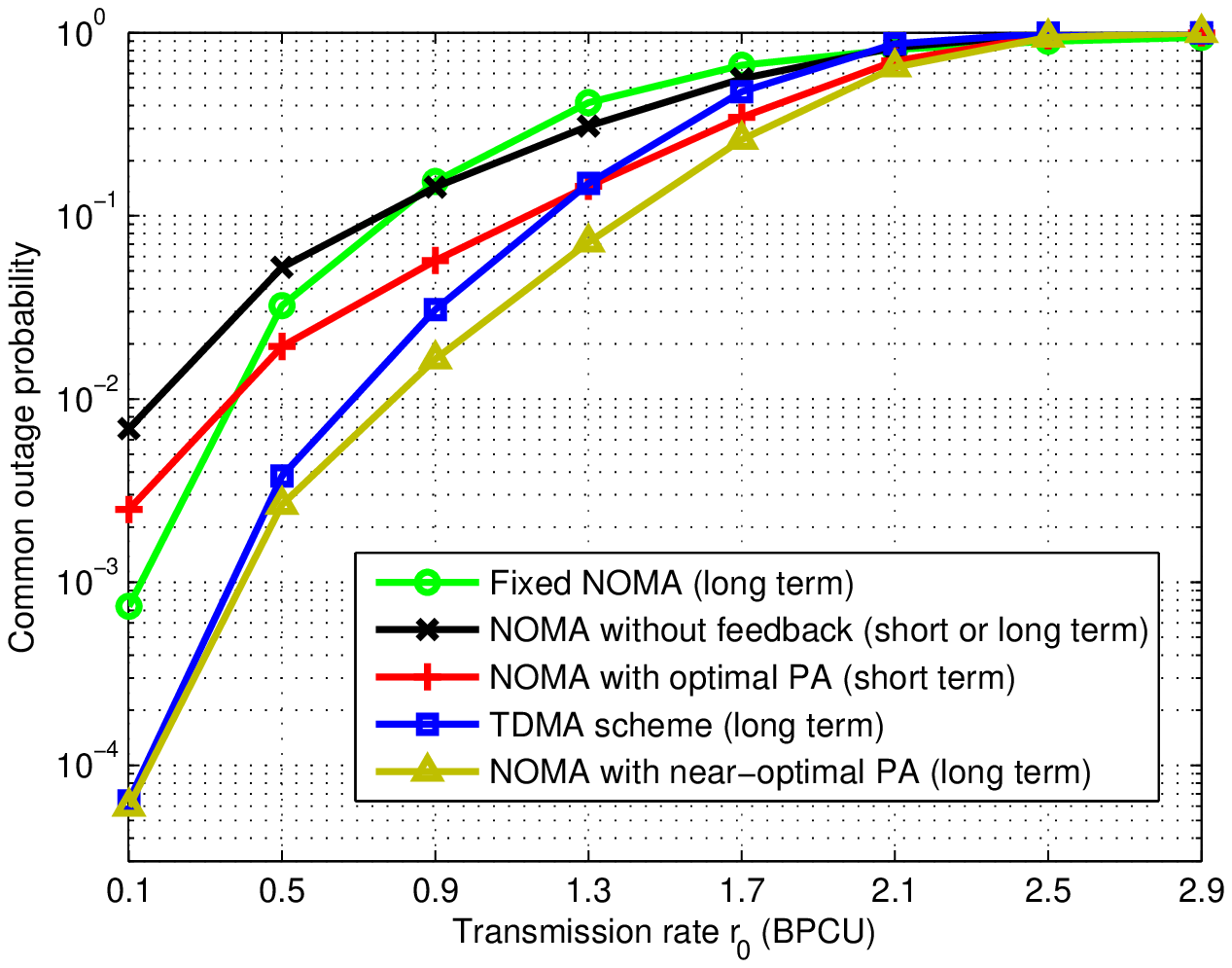}\vspace{-1em}
    \caption{COP versus transmission rate under the long-term power constraint, where the number of users is $K=3$, and the SNR is $20$ dB.}
    \label{fig:side:c1}
  \end{minipage}%
  \hspace{0.8cm}
  \begin{minipage}[t]{0.45\linewidth}
    \centering
    \includegraphics[width=3.3in]{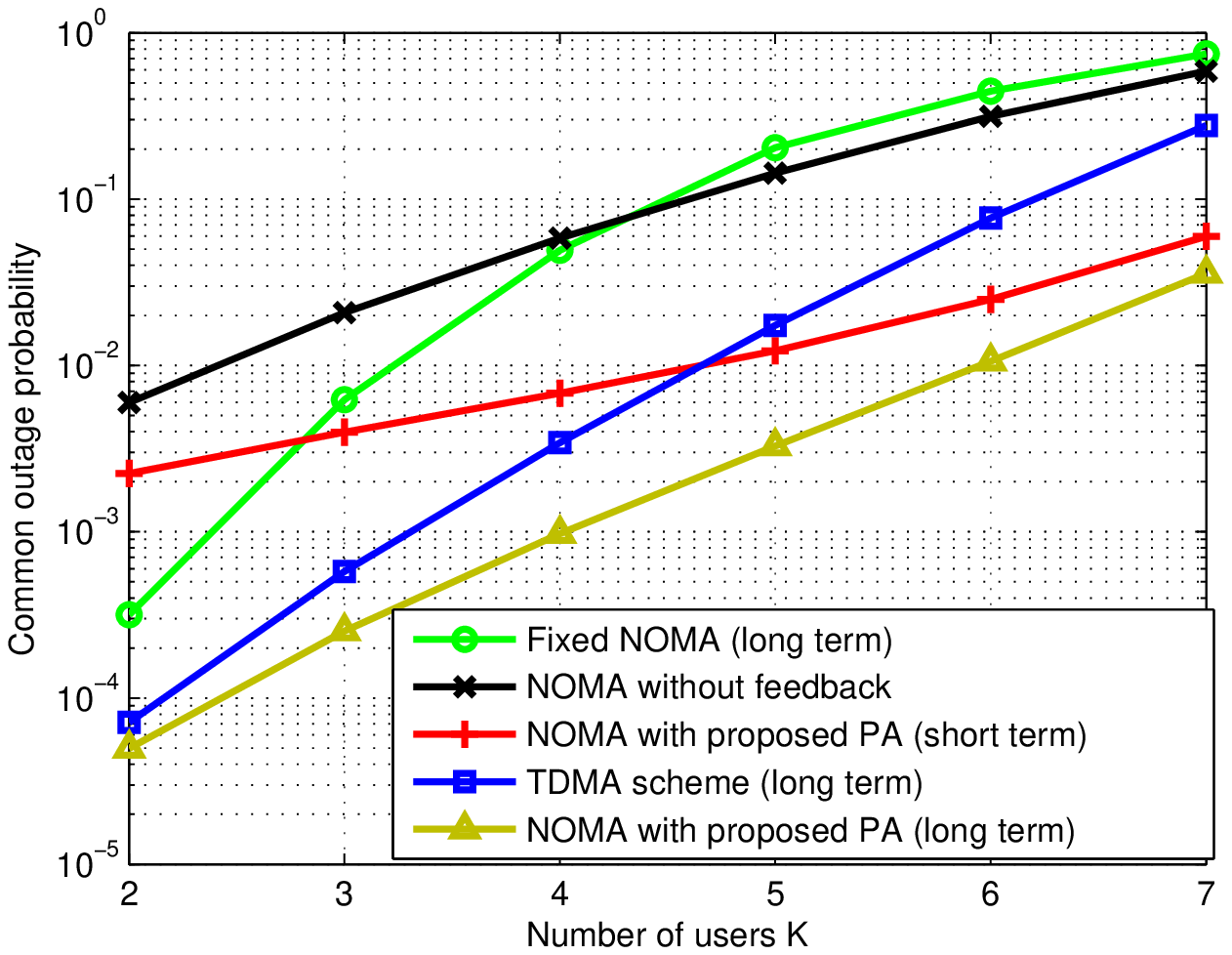}\vspace{-1em}
    \caption{COP versus the  number of users under the long-term power constraint, where  the target transmission rate is $r_0=1$ BPCU
    for each user, and  the SNR is $30$ dB.}
    \label{fig:side:c2}
  \end{minipage}
\end{figure}
\par}

This subsection  focuses on the outage performance of NOMA with one-bit feedback
 under the long-term power constraint in \eqref{problem_long}.
Figs. \ref{fig:side:b2}, \ref{fig:side:c1}, and \ref{fig:side:c2} compare the outage performance of
 NOMA with the   asymptotically optimal power allocation scheme proposed in Section \ref{subsection_long} with
  the benchmark schemes in Section \ref{subsection_benchmark}
 and NOMA under the short-term power constraint as a function of
 the SNR,   the transmission rate $r_0$, and the number of users  $K$,
respectively. As can be seen  in Fig. \ref{fig:side:b2}, under the long-term power constraint,
the COPs of NOMA with the proposed power allocation,
the TDMA scheme, and fixed NOMA  have the same slope at high SNR,
which is due to the fact that all these schemes achieve a diversity gain of $2$ (Lemma \ref{theorem_diversity_long}). However, fixed NOMA   suffers from   a poor  performance,
especially at high SNR. This implies that the power allocation scheme proposed in Section
\ref{subsection_long} plays an important role for improving  the outage performance.
Note that, although  the power allocation scheme proposed in Section
\ref{subsection_long} is based on the high-SNR approximation, it also  performs well at low SNR
compared to NOMA under the short-term power constraint.
As can be observed in Fig. \ref{fig:side:c1}, the fixed NOMA scheme also does not perform well especially
for large transmission rates $r_0$.
NOMA with the proposed asymptotically optimal long-term  power allocation scheme
has the best outage performance among the considered schemes.
When $r_0=1.3$, NOMA with the proposed power allocation scheme achieves a COP of approximate
  $0.07$, whereas the TDMA scheme
achieves only a COP of approximate  $0.15$.
 Finally, as shown in Fig. \ref{fig:side:c2},
the gap between the proposed NOMA scheme and the TDMA scheme increases  as $K$ increases.
The TDMA scheme with long-term power constraint  has  a   COP even higher than that of the NOMA  scheme
with short-term power constraint,
which means that the TDMA scheme is not suitable for scenarios with
large numbers of users due to its poor spectral efficiency.

{\renewcommand\baselinestretch{1}\selectfont
\begin{figure}[tbp]\centering \vspace{-1em}
    \epsfig{file=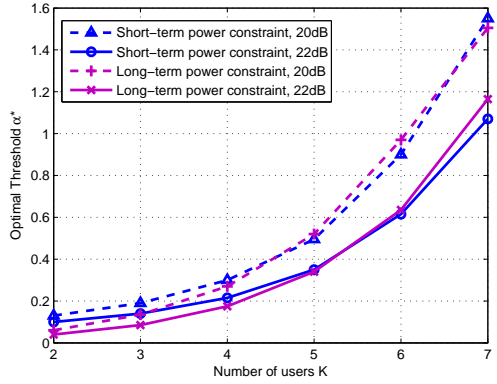, width=0.45\textwidth,clip=}
\caption{Optimal threshold versus the number of the users,
where  the target transmission rate is $r_0=1$ BPCU
    for each user, and  the SNR is $20$ dB or $22$ dB.}\label{Optima_alpha}
\vspace{-1em}
\end{figure}
\par}

Fig. \ref{Optima_alpha} illustrates the optimal threshold $\alpha^*$ versus the number
of users, $K$, where  the target transmission rate is $r_0=1$ BPCU
    for each user, and  the SNR is either $20$ dB or $22$ dB.
    As can be observed in  this figure, the optimal threshold increases
    significantly  with the number of users
and decreases with the SNR.
{ The optimal threshold decreases with the SNR for the following reason. Recall  that compared to the case of perfect CSI, the disadvantage of using one-bit feedback is that a user with a poor channel  may be categorized  as a user with a strong channel and hence given less transmit power. A good choice of $\alpha$ should avoid this problem as much as possible. For example, consider a scenario with two users, where the users' channels are ordered as $|h_1|^2\leq |h_2|^2$.  When the transmit power approaches infinity, one type of outage event is due to  the situation where users have very poor channel conditions, i.e., $|h_i|^2\rightarrow 0$, $i\in\{1,2\}$. In
this case, a good choice of $\alpha$ is $|h_1|^2\leq \alpha \leq |h_2|^2$, which means $\alpha\rightarrow 0$.
This intuition  can also be confirmed by the analytical results developed for the case with  the
 long-term power constraint. In particular,  Lemma 2 demonstrates that the maximum diversity gain can be achieved only when
   threshold $\alpha$ satisfies $\alpha\doteq P^{-1}$,
  i.e., the optimal threshold (denoted as $\alpha^*$) decreases with $P$ when $P$ is large.
Similarly, we can intuitively explain why the optimal threshold increases with the number of users $K$. Specifically, a small threshold $\alpha$ may result in a user $k$ with feedback ``1''
   having a poor channel, and thus, user $k$ with a poor channel
    may be mistakenly allocated with a large order index since the base station cannot distinguish the  channel gains
   with feedback ``1'' as discussed in Section \ref{ii}-A.
  Note that, when $K$ becomes large, the power allocated to
  a user   with a large order index  will become particularly small, according to the NOMA principle as discussed in \eqref{noma_constraint}.
   In this case, user $k$ with a poor channel will be given a very small amount of power,
   and thus  an outage event is prone  to happen.
 Therefore, $\alpha$ has to increase as  $K$ increases, in order to avoid this problem.
   }


\section{Conclusions}\label{vi}
This paper has investigated  the outage performance of downlink NOMA with 
  one-bit CSI feedback.
 We have derived a closed-form expression for the COP, as well as the optimal diversity gains
  under short-term and long-term power constraints.
 The diversity gain
 under the long-term power constraint was shown to be two whereas that under the short-term power constraint
 is only one.
 In order to minimize the COP, a dynamic power allocation policy based on the feedback state has also been proposed. For the short-term power constraint, we demonstrated that the original non-convex  problem can be transformed
 into a series  of convex problems.  For  the long-term power constraint,
 we have applied  high-SNR approximations to obtain an asymptotically optimal solution. Simulation results have been provided to demonstrate that the proposed NOMA schemes with one-bit feedback can outperform various existing multiple access schemes and achieve an outage performance close to the optimal one in many cases.
  An interesting topic for future research  is to extend the one-bit feedback scheme for NOMA  to multi-bit
  feedback. Moreover, the extension of the analysis of the one-bit feedback scheme
  to asymmetric scenarios with different distances and different rates for different users is also of interest.

\appendices
\section{Proof of Theorem \ref{theorem_cop}}\label{proof_theorem_cop}
  We first analyze the probability of  event $N=n$ defined in Definition \ref{Defini_event}, denoted
 by $\mathbb{P}_n(\alpha)$, which is a function of threshold $\alpha$. {Specifically, since all {\it unordered} channel gains are identically and independent distributed and $\mathbb{P}(Q(h_k)=0)=\mathbb{P}(|h_k|^2<\alpha)=
 1-e^{-\alpha}$, $\forall k\in[1:K]$,
the random variable $N$ defined in Definition \ref{Defini_event} is binomially distributed,
i.e., $N\sim B(K,1-e^{-\alpha})$. Thus,  $\mathbb{P}_n(\alpha)=C_K^n(1-e^{-\alpha})^ne^{-\alpha(K-n)}$
 as shown  in \eqref{Pr_n}.}

We then calculate the  outage probability  of individual  users for  event $N=n$, which is denoted by
$ \mathbb{P}_{k,n}^{\textrm{Indiv}}$ for user $\pi_k$. Note that an outage event at user $\pi_k$ occurs
if it fails to decode  the message for any user $\pi_l$,
$l\in[1:k]$.
 Therefore, the outage probability can be expressed as follows:
\begin{align}\label{Pr_i|n}
  \mathbb{P}_{k,n}^{\textrm{Indiv}}(\alpha,\mathbf{P}_n)&=1-\mathbb{P}\left\{\log(1+\textrm{SINR}_{l\rightarrow k,\;n})\geq r_0,\
  \forall l\in[1: k]\ \big|\ N=n\right\}\nonumber\\
 &=1- \mathbb{P}\left\{|{h}_{\pi_k}|^2 \geq \zeta_{l,n},\  \forall l\in[1:k]\ \big| \ N=n\right\}\nonumber\\
 &=\mathbb{P}(|{h}_{\pi_k}|^2\leq \hat{\zeta}_{k,n}\ \big|\ N=n).
\end{align}
Furthermore, based on  \eqref{CDF0} and \eqref{CDF1},
 $ \mathbb{P}_{k,n}^{\textrm{Indiv}}$ can be calculated as shown in \eqref{P_i|n}.

Moreover, the COP conditioned on event $N=n$, denoted as $\mathbb{P}_{n}^{\textrm{Common}}$,
can be obtained as follows:
\begin{align}\label{common_outage}
  \mathbb{P}_{n}^{\textrm{Common}}(\alpha,\mathbf{P}_n)&=1-\mathbb{P}
  \left\{\bigcap_{k\in[1:K]}\left\{\textrm{SINR}_{l\rightarrow k}\geq \hat{r}_0,
  \forall l\in[1:k]\right\}\ \big|\ N=n\right\}\nonumber\\
  &\stackrel{(a)}{=}1-
  \prod_{k=1}^K\mathbb{P}\left\{\textrm{SINR}_{l\rightarrow k}\geq \hat{r}_0,
  \forall l\in[1:k]\ \big|\ N=n\right\}\nonumber\\
  &=1-\prod_{k=1}^K (1-
\mathbb{P}_{k,n}^{\textrm{Indiv}}(\alpha,\mathbf{P}_n)),
\end{align}
where $(a)$ is due to the fact that, conditioned on  event $N=n$, the
 ${h}_{\pi_k}$'s are mutually independent as explained  in Remark \ref{Remark1},
and $\textrm{SINR}_{l\rightarrow k}$ is a function of ${h}_{\pi_k}$ as shown in \eqref{SNR}.

Now, the overall COP averaged over all  $(K+1)$ events can be expressed as
\begin{align}
   \mathbb{P}^{\textrm{Common}}(\alpha,\{P_{k,n}\})&=\sum_{n=0}^K  \mathbb{P}_n(\alpha)
  \mathbb{P}_{n}^{\textrm{Common}}(\alpha,\mathbf{P}_n).\label{common_outage_overall}
\end{align}
This completes the proof.

\section{Proof of Lemma \ref{theorem_diversity_short}}\label{proof_theorem_short}
\subsection{Proof of Achievability}
We will verify that a diversity gain of $1$ can be achieved based on a simple achievable power allocation scheme.
In particular, we set $\zeta_{k,n}=\frac{\mu_1}{P}$ in \eqref{zeta}, $\forall
k\in[1:K]$, $n\in[0:K]$,
where $\mu_1=(\hat{r}_0+1)^K-1$. Therefore, for any  $n$,
 $P_{k,n}=\frac{\hat{r}_0(\hat{r}_0+1)^{K-k}P}{(\hat{r}_0+1)^K-1}$ as shown in \eqref{Power-zeta}, and $\sum_{k=1}^K P_{k,n}=P$, i.e., the short-term power constraint is satisfied.
Using this power allocation, the outage probability in \eqref{P_i|n} can be expressed as:
 \begin{align}\label{Indiv_upper}
\mathbb{P}_{k,n}^{\textrm{Indiv}}=
  \left\{\begin{array}{ll}\min\left\{\frac{1-e^{-\frac{\mu_1}{P}}}{1-e^{-\alpha}}
  \thickapprox
  \frac{\mu_1}{P(1-e^{-\alpha})},1\right\}, &
  1\leq k\leq n;\\
  \left[1-e^{-(\frac{\mu_1}{P}-\alpha)}\right]^+, &n+1\leq k\leq K,\end{array}\right.
\end{align}
for a given $\alpha$.
Now, let $\alpha=\ln 2$, i.e., $e^{-\alpha}=\frac{1}{2}$ for simplicity.
 Then, from \eqref{Pr_n}, { $\mathbb{P}_n=\frac{C_K^n}{2^K}$}.
     From \eqref{Indiv_upper}, we have $\mathbb{P}_{k,n}^{\textrm{Indiv}}
    =0$ for $k\in [n+1:K]$ for a sufficiently large  $P$.
    So from \eqref{common_outage} and \eqref{Indiv_upper}, $\mathbb{P}_{n}^{\textrm{Common}}\thickapprox
    1-\left(1-\frac{\mu_1}{P(1-e^{-\alpha})}\right)^n\thickapprox
    \frac{2n\mu_1}{P}$ for $n\in[1:K]$, and
     $\mathbb{P}_{0}^{\textrm{Common}}=0$.  Thus, $\mathbb{P}^{\textrm{Common}}
     =\sum_{n=1}^K\mathbb{P}_n\mathbb{P}_{n}^{\textrm{Common}}
     \thickapprox \sum_{n=1}^K   { \frac{2nC_K^n\mu_1}{2^K P}}
     \doteq P^{-1}$     is obtained.

\subsection{Proof of Optimality}
Now, we  derive a lower bound on COP to verify that
the diversity gain of $1$ is optimal for all possible power allocations
 and all possible choices of  threshold  $\alpha$.
 From \eqref{zeta} and for the short-term power constraint, we have $\zeta_{k,n}\geq \frac{\hat{r}_0}{P}$, so
 $\mathbb{P}_{k,n}^{\textrm{Indiv}}$ can be lower bounded as:
    \begin{align}\label{Indiv_lower}
 \mathbb{P}_{k,n}^{\textrm{Indiv}}\geq
  \left\{\begin{array}{ll}\min\left\{\frac{1-e^{-\frac{\hat{r}_0}{P}}}{1-e^{-\alpha}},
  1\right\}, &
  1\leq k\leq n;\\
  \max\left[{1-e^{-(\frac{\hat{r}_0}{P}-\alpha)}}\right]^+, &n+1\leq k\leq K.\end{array}\right.
\end{align}
From \eqref{common_outage}, it can be observed that
\begin{align}\label{Common_lower} \mathbb{P}_n^{\textrm{Common}}
\geq  \mathbb{P}_{k,n}^{\textrm{Indiv}} \textrm{ for } \forall n\in[0:K], k\in[1:K].\end{align}
  Based on the above two relationships, in the following,
  we will verify that $\mathbb{P}^{\textrm{Common}}\dot{\geq} P^{-1}$
for any $\alpha$.
 Specifically, let $\alpha\doteq P^{\beta}$.

 First,  if $\beta>0$, from \eqref{Pr_n}, we have $\mathbb{P}_K\thickapprox1$. From  \eqref{Indiv_lower} and \eqref{Common_lower},
  $\mathbb{P}_{K}^{\textrm{Common}}\geq \frac{1-e^{-\frac{\hat{r}_0}{P}}}{1-e^{-\alpha}}
 \thickapprox \frac{\hat{r}_0}{P}\doteq P^{-1}$.
As shown in \eqref{common_outage_overall}, $\mathbb{P}^{\textrm{Common}}\geq
\mathbb{P}_{K}\mathbb{P}_{K}^{\textrm{Common}}\dot{\geq} P^{-1}$.

     Second, if $-1\leq\beta\leq 0$, from \eqref{Pr_n}, $\mathbb{P}_1
     \doteq P^{\beta}$.
     From \eqref{Indiv_lower} and \eqref{Common_lower},
      $\mathbb{P}_{1}^{\textrm{Common}}\geq \frac{1-e^{-\frac{\hat{r}_0}{P}}}{1-e^{-\alpha}}
      \thickapprox\frac{\hat{r}_0}{P(1-e^{-\alpha})}\doteq P^{-(1+\beta)}$ since $1-e^{-\alpha}\doteq P^{\beta}$.
     Thus, $\mathbb{P}^{\textrm{Common}}
     \geq\mathbb{P}_1\mathbb{P}_{1}^{\textrm{Common}}\dot{\geq} P^{-1}$.


     Finally, if $ \beta< -1$, from \eqref{Pr_n}, we have $\mathbb{P}_0\thickapprox1$.
     From  \eqref{Indiv_lower} and \eqref{Common_lower},
  $\mathbb{P}_{0}^{\textrm{Common}}\geq \frac{e^{-\alpha}-e^{-\frac{\hat{r}_0}{P}}}{e^{-\alpha}}
  \thickapprox \frac{\hat{r}_0}{P}-\alpha\doteq P^{-1}$. Thus,
  $\mathbb{P}^{\textrm{Common}}\geq
\mathbb{P}_{0}\mathbb{P}_{0}^{\textrm{Common}}\dot{\geq} P^{-1}$.

\section{Proof of Lemma \ref{theorem_diversity_long}}\label{proof_theorem_long}
\subsection{Proof of Achievability}
We will verify that a diversity gain  of $2$ can be achieved.
For a given $\alpha$, we consider a simple achievable power allocation scheme, i.e., we set $\zeta_{k,n}=\frac{\mu_1(K+1)
\mathbb{P}_n}{P}$ in \eqref{zeta}, $\forall k\in[1:K]$, $n\in[0:K]$,
 where $\mathbb{P}_n$ is given in \eqref{Pr_n}. 
This implies that    $P_{k,n}=\frac{\hat{r}_0(\hat{r}_0+1)^{K-k}P}{((\hat{r}_0+1)^K-1)(K+1)\mathbb{P}_n
}$, $k\in[1:K]$,
and $\sum_{k=1}^K P_{k,n}=\frac{P}{(K+1)\mathbb{P}_n}$.
The long-term power constraint in \eqref{power_long-term}  is obviously satisfied.
Using such power allocation, the outage probability in \eqref{P_i|n} can be expressed as:
 \begin{align}\label{Indiv_upper_long}
\mathbb{P}_{k,n}^{\textrm{Indiv}}=
  \left\{\begin{array}{ll}\min\left\{\frac{1-e^{-\frac{\mu_1(K+1)\mathbb{P}_n}{P}}}{1-e^{-\alpha}}
  \thickapprox
  \frac{\mu_1(K+1)\mathbb{P}_n}{P(1-e^{-\alpha})},1\right\}, &
  1\leq k\leq n;\\
  \left[1-{e^{-\left(\frac{\mu_1(K+1)\mathbb{P}_n}{P}-\alpha\right)}}\right]^+, &n+1\leq k\leq K.\end{array}\right.
\end{align}
Now, let $\alpha=\frac{\mu_1(K+1)}{P}\doteq P^{-1}$, so that $\frac{\mu_1(K+1)\mathbb{P}_n}{P}\leq \alpha$.
From \eqref{Pr_n}, $\mathbb{P}_n\thickapprox C_K^n \alpha^n\doteq P^{-n}$.
From \eqref{Indiv_upper_long}, $\mathbb{P}_{k,n}^{\textrm{Indiv}}\thickapprox
\frac{\mu_1(K+1)\mathbb{P}_n}{P\alpha}= \mathbb{P}_n$ for $k\in[1:n]$ and
$\mathbb{P}_{k,n}^{\textrm{Indiv}}=0$ for $k\in[n+1:K]$. Hence, from \eqref{common_outage},
$\mathbb{P}_{n}^{\textrm{Common}}\thickapprox 1-(1-\mathbb{P}_n)^n\thickapprox n\mathbb{P}_n$ for
$n\in[0:K]$. Furthermore, from \eqref{common_outage_overall},
$\mathbb{P}^{\textrm{Common}}\thickapprox\sum_{n=0}^K n(\mathbb{P}_n)^2\doteq
\sum_{n=1}^KP^{-2n}$, where $P^{-2}$ is the dominant term when $n=1$.

\subsection{Proof of Optimality}
Now, we  derive a lower bound on COP to verify that a diversity gain of $2$ is optimal under
the long-term power constraint.
 From \eqref{zeta} and the long-term power constraint, we have $\zeta_{k,n}\geq \frac{\hat{r}_0\mathbb{P}_n}{P}$, so
 $\mathbb{P}_{k,n}^{\textrm{Indiv}}$ can be lower bounded as:
    \begin{align}\label{Indiv_lower_long}
 \mathbb{P}_{k,n}^{\textrm{Indiv}}\geq
  \left\{\begin{array}{ll}\min\left\{\frac{1-e^{-\frac{\hat{r}_0\mathbb{P}_n}{P}}}{1-e^{-\alpha}},
  1\right\}, &
  1\leq k\leq n;\\
  \max\left\{\frac{e^{-\alpha}-e^{-\frac{\hat{r}_0\mathbb{P}_n}{P}}}{e^{-\alpha}},0\right\}, &n+1\leq k\leq K.\end{array}\right.
\end{align}
  Based on the above  relationship and \eqref{Common_lower},
  we can prove that $\mathbb{P}^{\textrm{Common}}\dot{\geq} P^{-2}$
for any $\alpha$.
 Specifically, let $\alpha\doteq P^{\beta}$.

 First,  if $\beta>0$, from \eqref{Pr_n}, we have $\mathbb{P}_K\thickapprox1$. From  \eqref{Indiv_lower} and \eqref{Common_lower},
  $\mathbb{P}_{K}^{\textrm{Common}}\geq \frac{1-e^{-\frac{\hat{r}_0\mathbb{P}_K}{P}}}{1-e^{-\alpha}}
 \thickapprox \frac{\hat{r}_0}{P}\doteq P^{-1}$.
As shown in \eqref{common_outage_overall}, $\mathbb{P}^{\textrm{Common}}\geq
\mathbb{P}_{K}\mathbb{P}_{K}^{\textrm{Common}}\dot{\geq} P^{-1}$.

     Second, if $-1\leq \beta\leq 0$, from \eqref{Pr_n}, { $\mathbb{P}_1=K(1-e^{-\alpha})e^{-\alpha(K-1)}
     \doteq P^{\beta}$}.
     From \eqref{Indiv_lower_long} and \eqref{Common_lower},
      $\mathbb{P}_{1}^{\textrm{Common}}\geq \frac{1-e^{-\frac{\hat{r}_0\mathbb{P}_1}{P}}}{1-e^{-\alpha}}
      \thickapprox\frac{\hat{r}_0\mathbb{P}_1}{P(1-e^{-\alpha})}=
      \frac{\hat{r}_0Ke^{-\alpha(K-1)}}{P}\doteq P^{-1}$.
     Thus, $\mathbb{P}^{\textrm{Common}}
     \geq\mathbb{P}_1\mathbb{P}_{1}^{\textrm{Common}}\dot{\geq} P^{\beta-1}\dot{\geq}P^{-2}$.


     Finally, if $ \beta< -1$, from \eqref{Pr_n}, we have $\mathbb{P}_0\thickapprox1$.
     From  \eqref{Indiv_lower_long} and \eqref{Common_lower},
  $\mathbb{P}_{0}^{\textrm{Common}}\geq \frac{e^{-\alpha}-e^{-\frac{\delta^2\hat{r}_0\mathbb{P}_0}{P}}}{e^{-\alpha}}
  \thickapprox \frac{\delta^2\hat{r}_0}{P}-\alpha\doteq P^{-1}$. Thus,
  $\mathbb{P}^{\textrm{Common}}\geq
\mathbb{P}_{0}\mathbb{P}_{0}^{\textrm{Common}}\dot{\geq} P^{-1}$.

Summarizing these three regions, 
the necessary condition to achieve the optimal
 diversity gain of $2$  is to set  $\beta=-1$.

 \section{Proof of Proposition \ref{poposition_P3}}\label{proof_proposition_P3}
For optimization problem (P3) in \eqref{new_problem_long2},
an asymptotically optimal solution $\{\zeta_{k,n}\}$ at high SNR has the following properties:
\begin{itemize}
  \item[(a)] when $n\geq 3$, $(\zeta_{1,n},\cdots,\zeta_{K,n})$ can be {\em any value} s.t. \eqref{constraint_long2} and \eqref{constraint_long22} and
 $ \zeta_{k,n}\dot{>} P^{-n}$; 
 \item[(b)]  $(\zeta_{1,2},\cdots,\zeta_{K,2})$ satisfy
 $P^{-2}\dot{<}\zeta_{k,2}\dot{<} P^{-1} $ for $k\in[1:K]$,
  $\zeta_{k,2}-\alpha\dot{<} P^{0}$ for $k\in[3:K]$; 
  \item[(c)] $(\zeta_{1,1},\cdots,\zeta_{K,1})$ satisfy $\zeta_{1,1}\dot{\leq} P^{-2}$,
  $\zeta_{k,1}-\alpha\dot{\leq} P^{-2}$, $\forall k\in[2:K]$;
  \item[(d)]$(\zeta_{1,0},\cdots,\zeta_{K,0})$ satisfy $\zeta_{k,0}-\alpha\dot{\leq} P^{-2}$, $\forall k\in[1:K]$.
\end{itemize}
\begin{proof}
From Lemma \ref{theorem_diversity_long}, we know that  the optimal threshold satisfies $\alpha\doteq P^{-1}$,
and the optimal COP satisfies $\mathbb{P}^{\textrm{Common}}\doteq P^{-2}$. These properties can be verified as follows.

(a) From \eqref{Pr_n} and \eqref{common_outage_overall}, we have
$\mathbb{P}_n\doteq P^{-n} $ for  $n\in[0:K]$. This implies
 that each term $\mathbb{P}_n  \mathbb{P}_{n}^{\textrm{Common}}$, $n\geq 3$, affects negligibly the
optimal  COP no matter what   power allocation scheme is used,
   and hence {\em any} power allocation scheme can be adopted
    when $n\geq 3$  as long as it consumes negligible power, i.e., $ \mathbb{P}_n
  \sum_{k=1}^K \frac{(\hat{r}_0+1)^{k-1}\hat{r}_0}{\zeta_{k,n}}
  \rightarrow 0$ as $P\rightarrow \infty$. Let $\zeta_{k,n}\doteq P^{\gamma_{k,n}}$,
  then  $ \mathbb{P}_n
  \sum_{k=1}^K \frac{(\hat{r}_0+1)^{k-1}\hat{r}_0}{\zeta_{k,n}}\doteq P^{-n-\min_{k\in[1:K]}\{\gamma_{k,n}\}}$,
  and hence $\min_{k\in[1:K]}\{\gamma_{k,n}\}>-n$. Combining this constraint with
  \eqref{constraint_long2} and \eqref{constraint_long22}, property (a) is verified.

  (b) To verify property (b), we will show that  $\mathbb{P}_2  \mathbb{P}_{2}^{\textrm{Common}}\dot{<} P^{-2}$
   (i.e., $\mathbb{P}_2  \mathbb{P}_{2}^{\textrm{Common}}$ is negligible compared to the optimal
  COP) can be achieved at  negligible power cost  for the term $ \mathbb{P}_2
  \sum_{k=1}^K \frac{(\hat{r}_0+1)^{k-1}\hat{r}_0}{\zeta_{k,2}}$, only when  $\{\zeta_{k,2}\}$ satisfies the constraints in property (b). Let $\zeta_{k,2}\doteq P^{\gamma_{2,k}}$,
  then, similar to the proof of property (a), $\min_{k\in[1:K]}\{\gamma_{k,2}\}>-2$ should be
  satisfied  such that $\mathbb{P}_2
  \sum_{k=1}^K \frac{(\hat{r}_0+1)^{k-1}\hat{r}_0}{\zeta_{k,2}}\rightarrow 0$ as $P\rightarrow \infty$.
  Moreover, to achieve $\mathbb{P}_2  \mathbb{P}_{2}^{\textrm{Common}}\dot{<} P^{-2}$,
    $\mathbb{P}_{k,2}^{\textrm{Indiv}}\dot{<} P^{0}$ needs to be satisfied according to \eqref{Pr_n} and \eqref{common_outage}, $\forall k\in[1:K]$. Thus, according to
    \eqref{P_i|n},
    $\zeta_{k,2}\dot {<} P^{-1}$ can be verified  for $k\in[1:2]$, with the choice of $\alpha\doteq P^{-1}$;  and $\zeta_{k,2}-\alpha\dot {<} P^{0}$ for $k\in[3:K]$.

    (c) To achieve $\mathbb{P}^{\textrm{Common}}\doteq P^{-2}$,
     $\mathbb{P}_1  \mathbb{P}_{1}^{\textrm{Common}}\dot{\leq} P^{-2}$ has to be satisfied.
     Thus, $\mathbb{P}_{k,1}^{\textrm{Indiv}}\dot{\leq} P^{-1}$, $\forall k\in[1:K]$, needs to be satisfied according to \eqref{Pr_n} and \eqref{common_outage}.  Thus, with the choice  $\alpha\doteq P^{-1}$,
     property (c) can be verified based on    \eqref{P_i|n}.

(d) Similar to the proof of property (c), property (d) can be verified. The details are omitted here
 for brevity.
\end{proof}

From property (a), we know that more than one asymptotically optimal solution exists. To unify
the expression of the approximation, for $n\geq 3$,
 we set $\zeta_{k,n}$ to satisfy $\zeta_{k,n}\dot{<}P^{-1}$ and
$\zeta_{k,n}-\alpha\dot{<} P^{0}$ without loss of asymptotic optimality.
 Thus, together with properties (b)-(d), $\forall n\in[0:K]$, we have
 $\frac{\zeta_{k,n}}{1-e^{\alpha}}\ll 1$ for $k\in[1:n]$ and
 $\zeta_{k,n}-\alpha\ll 1 $ for $k\in[n+1:K]$.
Now, the following approximation can be obtained:
 \begin{align}
    \frac{e^{-{\zeta}_{k,n}}-e^{-\alpha}}
{1-e^{-\alpha}}=1-\frac{1-e^{-\zeta_{k,n}}}{1-e^{-\alpha}}\thickapprox 1-
 \frac{{{\zeta}_{k,n}}}{1-e^{-\alpha}}\thickapprox e^{-
  \frac{{{\zeta}_{k,n}}}
{1-e^{-\alpha}}}.
 \end{align}
 Accordingly, using a Taylor series expansion,
  the approximation of  function $f_{3,n}$ can be expressed as follows:
\begin{align}f_{3,n}(\alpha,\boldsymbol{\zeta}_n)&\thickapprox 1-\prod_{k=1}^n e^{-
  \frac{{{\zeta}_{k,n}}}
{1-e^{-\alpha}}}\prod_{k=n+1}^K e^{-({\zeta}_{k,n}-\alpha)}\nonumber\\
&=1-e^{-\left(\sum_{k=1}^n \frac{{{\zeta}_{k,n}}}
{1-e^{-\alpha}}+\sum_{k=n+1}^K({\zeta}_{k,n}-\alpha)\right)}\nonumber\\
&\thickapprox\sum_{k=1}^n \frac{{{\zeta}_{k,n}}}
{1-e^{-\alpha}}+\sum_{k=n+1}^K({\zeta}_{k,n}-\alpha).\end{align}
With this, problem (P3) in \eqref{new_problem_long2} has been approximately transformed to (P4) in \eqref{P3}.

\section{Proof of Proposition \ref{poposition_zeta2}}\label{proof_theorem_longpower}
This proposition can be proved by using \eqref{kkt2} and \eqref{P4_zeta}.
   For a given $n$ and depending on the values of $k$, three cases need to be considered. Firstly, when $k\in[1:n-1]$, from \eqref{P4_zeta},  $\zeta_{k,n}\leq\zeta_{k+1,n}$ holds   since   $\hat{r}_0>0$.
   Secondly, when $k=n$,
   we have $$ \zeta_{n,n}=\sqrt{\omega (\hat{r}_0+1)^{n-1}\hat{r}_0(1-e^{-\alpha})}
    <\sqrt{\omega (\hat{r}_0+1)^{n}\hat{r}_0} \leq \sqrt{\frac{\omega \mathbb{P}_n(\alpha)(\hat{r}_0+1)^{n}\hat{r}_0}{\mathbb{P}_n(\alpha)-\lambda_{n+1,n}}}=\zeta_{n+1,n}$$
     since $\lambda_{n+1,n}   \geq 0$.
     Thirdly, when $k\in[n+1:K-1]$,  two subcases with respect to
      $\lambda_{k,n}$ are considered.  If $\lambda_{k,n}>0$, $\zeta_{k,n}=\alpha$
      can be obtained from \eqref{kkt2}, so $\zeta_{k,n}\leq\zeta_{k+1,n}$ holds since
      $\zeta_{k+1,n}\geq \alpha$. If $\lambda_{k,n}=0$,   since $\lambda_{k+1,n}   \geq 0$, we have
      $$\zeta_{k,n}=
      \sqrt{\omega (\hat{r}_0+1)^{k-1}\hat{r}_0}\leq
      \sqrt{\frac{{ \omega}\mathbb{P}_n(\alpha)(\hat{r}_0+1)^{k}\hat{r}_0}
      {\mathbb{P}_n(\alpha)-\lambda_{k+1,n}}}=\zeta_{k+1,n}.$$ This  completes the proof.



\section{Proof of Theorem \ref{theorem_algorithm}}\label{proof_theorem_algorithm}
 During  the $t$-th
   iteration,  $\omega^{(t)}$, $\lambda_{k,n}^{(t)}$, and $\zeta_{k,n}^{(t)}$ are calculated according to
      \eqref{omega}, \eqref{lambda}, and \eqref{eq_theorem_long}, respectively.
   Assume   that $i_n^*\geq i_n^{(t)}$,  $\forall n\in[0:K]$, and
  the constraints in \eqref{kkt_constraint} are not satisfied,
     i.e.,  we have to further enlarge
   at least one $i_n^{(t)}$ to find $\{i_n^*\}$. 
Now, divide $\{n\}$ into two sets:
 \begin{align}\label{two_sets}\textrm{$\mathcal{N}_{1}^{(t)}\triangleq \{n:\zeta_{n+i_n^{(t)}+1,n}^{(t)}>\alpha\}$
   and  $\mathcal{N}_{2}^{(t)}\triangleq \{n:\zeta_{n+i_n^{(t)}+1,n}^{(t)}\leq\alpha\}$.}
   \end{align}
 According to the definitions  in \eqref{two_sets},
 we first present an important proposition as follows.
   \begin{Proposition}\label{poposition_omega}
     For the $(t+1)$-th iteration, $\omega^{(t+1)}>\omega^{(t)}$ if we enlarge any   $i_n^{(t)}$ with
     $n\in\mathcal{N}_{1}^{(t)}$;
     $\omega^{(t+1)}\leq \omega^{(t)}$ if we enlarge any   $i_n^{(t)}$ with $n\in\mathcal{N}_{2}^{(t)}$.
   \end{Proposition}
   \begin{proof}
    This proposition can be proved based on \eqref{eq_theorem_long} and \eqref{omega}.
   For the $t$-th iteration, we first consider the case that some   $i_n^{(t)}$
   with $n\in\mathcal{N}_{1}^{(t)}$
   is selected to be enlarged in the next iteration,
   where   $\mathcal{N}_{1}^{(t)}$ is defined in \eqref{two_sets}.
   Assume without loss of generality that $i_{m}^{(t)}$ with $m\in \mathcal{N}_{1}^{(t)}$ is enlarged, i.e.,   $i_{m}^{(t+1)}=i_{m}^{(t)}+l$, $l\in[1:K-n-i_n^{(t)}]$, and the other $i_n^{(t)}$'s remain unchanged,
   i.e., $i_{n}^{(t+1)}=i_{n}^{(t)}$, $\forall n\neq m$.
   According to \eqref{omega}, we have
   \begin{align}
     \frac{1}{\sqrt{\omega^{(t)}}}{\sum_{n=0}^K\mathbb{P}_n(\alpha)A_n(i_n^{(t)})}=
  {P-\sum_{n=0}^K\mathbb{P}_n(\alpha)B_n (i_n^{(t)})}.
   \end{align}
   Let $u_2(k)\triangleq{(\hat{r}_0+1)^{k-1}\hat{r}_0}$. From \eqref{An} and \eqref{Bn},
   the above equality can be rewritten as
    \begin{align}
     &\frac{1}{\sqrt{\omega^{(t)}}}\left[\sum_{k=1}^m \frac{\sqrt{u_2(k)}}{1-e^{-\alpha}}
     +\sum_{k=m+i_m^{(t)}+l+1}^K\sqrt{u_2(k)}+
     {\sum_{n=0,n\neq m}^K\mathbb{P}_n(\alpha)A_n(i_n^{(t)})}\right]\nonumber\\
     =&
  P-\sum_{n=0}^K\mathbb{P}_n(\alpha)B_n (i_n^{(t)})-\sum_{k=m+i_m^{(t)}+1}^{m+i_m^{(t)}+l}
  \frac{u_2(k)}{\sqrt{w^{(t)}u_2(k)}}. \label{eq_appendix1}
   \end{align}
   Similarly, for the $(t+1)$-th iteration, since only $i_m^{(t)}$ is enlarged,  we obtain
    \begin{align}
     &\frac{1}{\sqrt{\omega^{(t+1)}}}\left[\sum_{k=1}^m \frac{\sqrt{u_2(k)}}{1-e^{-\alpha}}
     +\sum_{k=m+i_m^{(t)}+l+1}^K\sqrt{u_2(k)}+
     {\sum_{n=0,n\neq m}^K\mathbb{P}_n(\alpha)A_n(i_n^{(t)})}\right]\nonumber\\
     =&
  P-\sum_{n=0}^K\mathbb{P}_n(\alpha)B_n (i_n^{(t)})-\sum_{k=m+i_m^{(t)}+1}^{m+i_m^{(t)}+l}
  \frac{u_2(k)}{\alpha}.\label{eq_appendix2}
   \end{align}
   Comparing the right hand side terms of \eqref{eq_appendix1} and \eqref{eq_appendix2},
   the one in \eqref{eq_appendix1} is larger than the one in \eqref{eq_appendix2}, since
   $\zeta_{k,n}^{(t)}=\sqrt{w^{(t)}u_2(k)}>\alpha$, $\forall k\in[m+i_m^{(t)}+1:
   m+i_m^{(t)}+l]$, which can be obtained from \eqref{eq_theorem_long} and the
   definition of $\mathcal{N}_1^{(t)}$ in \eqref{two_sets}. Thus,
   $\frac{1}{\sqrt{w^{(t)}}}>\frac{1}{\sqrt{w^{(t+1)}}}$ can be obtained by comparing the left hand
   side terms of \eqref{eq_appendix1} and \eqref{eq_appendix2}.

Now, we have proven  that $w^{(t+1)}>w^{(t)}$ if  we enlarge any   $i_n^{(t)}$ with $n\in\mathcal{N}_{1}^{(t)}$. Following similar steps, we can show that
     $\omega^{(t+1)}\leq \omega^{(t)}$ if we enlarge any   $i_n^{(t)}$ with $n\in\mathcal{N}_{2}^{(t)}$.
       \end{proof}

     Based on Proposition \ref{poposition_omega}, another proposition is given in the following.
     \begin{Proposition}\label{proposition_prin1}
      For the $t$-th iteration, at least one $i_n^{(t)}\neq i_n^*$ for
      $n\in\mathcal{N}_{2}^{(t)}$ must exist.
     \end{Proposition}
     \begin{proof}
   Reduction to absurdity is adopted.  We first assume that $i_n^{(t)}=i_n^*$, $\forall n\in \mathcal{N}_{2}^{(t)}$, and we need to find the other $i_n^*$'s by enlarging at least one $i_n^{(t)}$
   with $n\in\mathcal{N}_{1}^{(t)}$.
  According to \eqref{eq_theorem_long},  $\zeta_{n+i_n^{(t)}+1,n}^{(t)}=
 \sqrt{w^{(t)}(\hat{r}_0+1)^{i_n^{(t)}}\hat{r}_0}>\alpha$ when   $n\in\mathcal{N}_{1}^{(t)}$;
 thus, it is easy to obtain $\lambda_{n+i_n^{(t+1)}+1,n}<0$ if we enlarge any   $i_n^{(t)}$ with
 $n\in\mathcal{N}_{1}^{(t)}$,
  based on \eqref{lambda}
 and Proposition \ref{poposition_omega}. In addition, $\mathcal{N}_{1}^{(t+1)}
 =\mathcal{N}_{1}^{(t)}$ obviously holds, i.e., the set $\mathcal{N}_{1}^{(t)}$ will not change
 in the next iteration.
 By analogy,
 $\lambda_{n+i_n^{(t')}+1,n}<0$, $\forall t'\geq t+1$ if we enlarge any $i_n^{(t)}$ with
  $n\in\mathcal{N}_{1}^{(t)}$.
  This implies  that the constraint in \eqref{kkt_constraint}
 would never  be satisfied if  $i_n^{(t)}=i_n^*$,  $\forall n\in \mathcal{N}_{2}^{(t)}$.
 \end{proof}
According to Proposition \ref{proposition_prin1},
$\{i_n^*\}$ must be found using the following update  rule.

{\bf Rule 1}: Enlarge at least one element $i_n^{(t)}$ with $n\in\mathcal{N}_2^{(t)}$ in the $t$-th iteration.

In order to improve the search efficiency, Rule 1 can be further refined into another update rule.
A proposition is first given as follows.
   \begin{Proposition}\label{proposition_prin2}
   When Rule 1 is adopted for searching  $\{i_n^*\}$,  $i_n^*\geq v^{(t)}$,  $\forall n\in\mathcal{N}_{2}^{(t)}$,  where $v^{(t)}\triangleq\arg\max_{i\in\left[i_n^{(t)}:K-n\right]}
   \{i:\zeta_{n+i}^{(t)}\leq \alpha\}$.
   \end{Proposition}
   \begin{proof}
   Also using reduction to absurdity, we first assume that
   there exists one $i_n^*\in\left[i_n^{(t)}+1:v^{(t)}-1\right]$ with
  $n\in\mathcal{N}_{2}^{(t)}$. Then, $\zeta_{n+i_n^*+1,n}^{(t)}=
  \sqrt{w^{(t)}(\hat{r}_0+1)^{i_n^{(t)}}\hat{r}_0}\leq \alpha$
   according to \eqref{eq_theorem_long} and \eqref{two_sets}; thus, it is easy to obtain
   $\zeta_{n+i_n^*+1,n}^{(t+1)}\leq \alpha$ based on
    Proposition \ref{poposition_omega} and Rule 1.
  By analogy, $\zeta_{n+i_n^*+1,n}^{(t')}\leq \alpha$, $\forall t'\geq t+1$ when enlarging at least one element $i_n^{(t)}$ with $n\in\mathcal{N}_2^{(t)}$ (i.e., Rule 1).
  This implies  that the constraint in \eqref{kkt_constraint}
 would never  be satisfied if there existed any $i_n^*\in\left[i_n^{(t)}+1:v^{(t)}-1\right]$ with
 $n\in\mathcal{N}_{2}^{(t)}$.
 \end{proof}
 Now, according to Proposition \ref{proposition_prin2}, Rule 1 can be refined into Rule 2 to further
 improve search efficiency as follows.

 {\bf Rule 2}:
 Enlarge each $i_n^{(t)}$ with $i_n^{(t)}<K-n$ and $n\in\mathcal{N}_{2}^{(t)}$ as
 $i_n^{(t+1)}=v^{(t)}$. 

 Based on Propositions \ref{proposition_prin1} and \ref{proposition_prin2},
  $\{i_n^*\}$ must be found using Rule 2.
Note that Rule 2 has been adopted in Step 2-c of Algorithm I (Section \ref{subsection_long}), and Theorem \ref{theorem_algorithm} is proved.

 \begin{Remark} Using  Rule 2 in Step 2-c of Algorithm I, we can easily verify that the constraint
 $\lambda_{k,n}^{(t)}\geq 0$ always holds
    for $k\in[n+1:n+i_n^{(t)}]$, according to \eqref{eq_theorem_long}, \eqref{lambda}, and Proposition \ref{poposition_omega}.
    Thus, it is not necessary to include this constraint in Step 2-b of Algorithm I.
    \end{Remark}

\bibliographystyle{IEEEtran}
\bibliography{references}
\end{document}